\pdfoutput=1

\documentclass[english,11pt]{article}
\usepackage[utf8]{inputenc}
\usepackage[T1]{fontenc}
\usepackage[margin=1in]{geometry}
\usepackage{microtype}
\usepackage{hyphenat}
\usepackage[left]{lineno} 

\usepackage{mathtools} 
\usepackage{amssymb}   
\usepackage{amsthm}
\usepackage{bm}
\usepackage{dsfont}
\usepackage{mathrsfs}
\usepackage{nicefrac}

\usepackage{xcolor}
\usepackage{graphicx}
\usepackage{float}
\usepackage{verbatim}
\usepackage{pgfplots}
\usepackage{array}
\usepackage{tabularx}
\usepackage{booktabs}
\usepackage{multirow}
\usepackage{makecell}
\usepackage{arydshln}
\usepackage[font=small]{caption}
\usepackage{subcaption}
\usepackage[shortlabels]{enumitem} 
\usepackage[algo2e]{algorithm2e}   
\usepackage{algorithmic}          

\usepackage{xr}
\usepackage[authoryear]{natbib}
\usepackage[colorlinks,citecolor=blue,urlcolor=blue]{hyperref}
\usepackage{cleveref} 

\setlist[itemize]{leftmargin=1em}
\setlist[enumerate]{leftmargin=1em}
\allowdisplaybreaks
\newcolumntype{L}{>{\raggedright\arraybackslash}X}
\newcolumntype{C}{>{\centering\arraybackslash}m{1.2cm}}
\newcolumntype{N}{>{\centering\arraybackslash}m{0.9cm}}

\floatstyle{ruled}
\newfloat{algorithm}{tbp}{loa}
\newcommand{\algorithmname}{Algorithm}
\floatname{algorithm}{\protect\algorithmname}

\makeatletter
\newcommand{\blfootnote}[1]{%
  \begingroup
  \renewcommand{\@thefnmark}{}%
  \@footnotetext{#1}%
  \endgroup
}
\newcommand{\myCref}[1]{\@ifundefined{r@#1}{\Cref*{ec:#1}}{\Cref{#1}}}
\newcommand{\myref}[1]{\@ifundefined{r@#1}{\ref*{ec:#1}}{\ref{#1}}}
\makeatother

\newcommand{\cC}{\mathcal{C}}
\newcommand{\cD}{\mathcal{D}}
\newcommand{\cE}{\mathcal{E}}

\newcommand{\cG}{\mathcal{G}}

\newcommand{\cP}{\mathcal{P}}

\newcommand{\cZ}{\mathcal{Z}}

\newcommand{\mE}{\mathbb{E}}

\newcommand{\mP}{\mathbb{P}}

\newcommand{\eq}{\[}
\newcommand{\en}{\]}

\newcommand{\D}{\mathcal{D}}
\newcommand{\Pop}{\mathcal{P}}   
\newcommand{\Th}{\Theta}

\newcommand{\sce}{\psi}
\newcommand{\scesp}{\Psi}

\newcommand{\Loss}{L}
\newcommand{\disc}[2]{\Loss\!\left(#1,#2\right)}

\newcommand{\Cj}[1]{\mathcal{C}_j\!\left(#1\right)}
\newcommand{\Vhatm}[1]{\hat V_m\!\left(#1\right)}

\newcommand{\iid}{i.i.d.\ }

\theoremstyle{plain}
\newtheorem{lemma}{Lemma}[section]

\newtheorem{theorem}{Theorem}[section]
\newtheorem{corollary}{Corollary}[section]
\newtheorem{assumption}{Assumption}[section]
\newtheorem{example}{Example}[section]

\theoremstyle{definition}
\newtheorem{remark}{Remark}

\theoremstyle{plain}

\numberwithin{equation}{section}

\usepackage[left]{lineno} 

\begin{document}

\theoremstyle{plain}

\title{\Large Model-Free Assessment of Simulator Fidelity via Quantile Curves}

\author{%
  Yu\mbox{-}Shiou Willy Lin\textsuperscript{1}
  \qquad
  Kaizheng Wang\textsuperscript{1,2}
  \qquad
  Garud Iyengar\textsuperscript{1,2}
  \\[0.5em]
  \textsuperscript{1}Department of Industrial Engineering
  and Operations Research\\
  \textsuperscript{2}Data Science Institute\\
  Columbia University%
  \thanks{Email addresses:
    \texttt{yl5782@columbia.edu},
    \texttt{kaizheng.wang@columbia.edu},
    \texttt{garud@ieor.columbia.edu}.}
}
\date{}

\maketitle

\begin{abstract}
As generative AI models are increasingly used to simulate real-world systems, quantifying their ``sim-to-real'' gap is critical. We study this gap across a population of input settings, called \emph{scenarios}, e.g. survey questions or operating conditions. The population-level quantities defining the  real-world and simulated system are only observed through finite, often heterogeneous, samples. Consequently, the population-level discrepancy cannot be directly computed, and standard predictive inference methods that target observable outputs are ill-suited. We propose a model-agnostic framework that constructs confidence sets for the latent parameters, forms a conservative proxy for the sim-to-real discrepancy, and estimates its quantile function across scenarios. The resulting calibrated risk profile allows for inference on a new scenario, tail-risk summaries such as Conditional Value-at-Risk (CVaR), and principled comparisons across simulators. Our method applies to  a range of output spaces, including categorical survey responses and continuous outcomes. We demonstrate its utility by evaluating the alignment of four major LLMs with human populations on the WorldValueBench dataset.
\end{abstract}


\section{Introduction}

The recent advances in artificial intelligence (AI) have accelerated the
adoption of simulators across operations and manufacturing, 
agent-based modeling in social-science research, user surveys, and
education
\citep{ZHANG2019103123,abms,Argyle_Busby_Fulda_Gubler_Rytting_Wingate_2023,human_sub}. 
At
the platform level, industry ecosystems such as NVIDIA Omniverse and
Earth-2 exemplify the push toward high-fidelity, AI-enabled digital twins,
and LLMs are increasingly used to build generative agents that emulate
human responses \citep{park2023generative, lu2025prompting}. Collectively,
these developments are driving  growth in simulation across domains.  

Against this backdrop, quantifying the \emph{distributional} discrepancy
between simulated and real-world, or ground truth, outcomes is critical. This
``sim-to-real'' gap is a pervasive challenge in domains ranging from LLMs
to computing systems and robotics~\citep{llm_discrepancy,
  durmus2023towards, tobin2017domain, peng2018sim}. Crucially, the gap depends
on the particular \emph{scenario} (e.g., operating conditions or survey
questions). 
Thus, 
the gap in a particular 
scenario 
is a sample from 
\emph{latent scenario-level 
  distribution} for both the real and simulated systems. 
Assessing fidelity therefore
requires characterizing this discrepancy distribution, rather than a
single average error, which can hide rare but consequential failures. This
distributional characterization is the primary objective of our work. 

Existing uncertainty quantification methods are insufficient for this
  task since they focus on assessing \emph{individual} outcomes
  (e.g.,~conformal prediction), and do not 
  quantify the \emph{distribution} of the discrepancy between 
  the simulator
  and real system across different scenarios.
  Conflating outcome validity
  with distributional fidelity may harm downstream applications. For
  instance, in survey research, an LLM-based simulator might generate  coherent individual responses while yielding systematically biased
  \emph{population-level  aggregates}~\citep{westwood2025surveyexistential}, and recent work has
  identified distributional alignment as a key
  objective~\citep{cao2025survey}. We provide theoretical guarantees for
  sim-to-real discrepancy assessment without making structural assumptions
  about the simulator, treating it strictly as a black box. 

\paragraph{Contributions}
We characterize the quantile function of the discrepancy between real
world  and
simulated output distributions across possible scenarios. Since the
real-world
distribution is observed only through finite, often heterogeneous samples,
the 
true quantile function 
cannot be computed. 
We construct per-scenario confidence sets for the latent real-world parameters, use them to form a proxy for the sim-to-real discrepancy, and establish finite-sample guarantees for the quantile curve at any desired level. Moreover, we allow the coverage of the confidence set 
to be chosen as a function of the number of available samples. This yields an adaptive procedure that produces tighter quantile curves in data-rich scenarios, while preserving reliability and interpretability when data are scarce. 

We illustrate our methodology using the dataset \emph{WorldValueBench}~\citep{zhao2024worldvaluesbenchlargescalebenchmarkdataset}, constructed from the World Values Survey~\citep{WVS_Wave7_2020}, by generating synthetic response profiles and benchmarking four LLMs against human survey data. Leveraging the large human sample size, we also study tightness and propose a data-driven procedure to construct confidence bands for the population quantile function.

Our main contributions are:
\begin{enumerate}
\item \textbf{Model-free, finite-sample fidelity profiling.} A black-box assessment framework with finite-sample, distributional guarantees under heterogeneous sample sizes across scenarios.
\item \textbf{Quantile-based risk summaries.} A calibrated discrepancy quantile curve enabling the estimation of tail-risk measures, such as CVaR, and other summaries.
\item \textbf{Inference and comparison.} Statistical inference for real-world parameters in new scenarios and principled comparisons across simulators.
\end{enumerate}

\paragraph{Related Literature}


The sim-to-real discrepancy has been a central concern in the uncertainty quantification (UQ) literature. 
Following the taxonomy of \citet{ROY20112131}, we distinguish \emph{input} uncertainty from \emph{output} uncertainty. Work on input uncertainty \citep{chen2024quantifying,barton2014quantifying, lam2022cheap} typically considers a simulator with parameters learned from finite real-world input data, and studies how the uncertainty of learned parameters, together with Monte Carlo noise, propagates to output quantities. These methods usually need refitting on bootstrap samples of the input data. Our setting is different: modern simulators such as LLMs or digital twins are often black boxes and may be systematically misspecified. Moreover, we do not have the training data, and it can be computationally prohibitive to refit. We therefore study \emph{output-level fidelity}: the discrepancy between real and simulated output distributions across a population of scenarios. This is related to UQ using simulation output \citep{jeon2024uncertaintyquantificationusingsimulation} and simulation-based inference \citep{cranmer2020frontier,pmlr-v89-papamakarios19a}, but our target is not Monte Carlo error or posterior inference on simulator parameters; it is a finite-sample assessment of scenario-level sim-to-real discrepancy. Two-sample methods are also sample-based and model-agnostic \citep{gretton2012kernel,JMLR:v24:21-1289}. They can provide within-scenario discrepancy measures and confidence sets through suitable test inversion. Our target, however, is the distribution of discrepancies across scenarios, where each discrepancy involves a latent real-world parameter observed only through a finite batch of outcomes. Our contribution is, therefore, by adopting a two-level procedure: construct observable one-sided proxies within each scenario, then calibrate their quantiles across scenarios, yielding finite-sample fidelity profiles and inference for latent parameters in new scenarios.

A second line of work studies LLM-based human simulation and opinion alignment, asking whether synthetic responses reproduce human responses \citep{llm_opinion,huang2025uncertaintyquantificationllmbasedsurvey}. Existing approaches often report aggregate alignment metrics or confidence sets for specific population quantities. Our goal is different: rather than summarizing fidelity by a single metric or validating one estimand, we estimate a quantile curve that describes how simulator error varies across scenarios.

Finally, our guarantees are related to conformal and predictive inference
\citep{vovk_2005,bates2021distribution}, which provide finite-sample,
distribution-free guarantees for black-box predictors. Standard conformal
methods target future \emph{observable} outcomes, whereas our target is a
\emph{scenario-level distributional parameter} observed only through finite,
often heterogeneous batches. Recent conformal variants for distributional objects \citep{snell2022quantile,budde2025statistical} are closer in spirit, but do not directly handle latent scenario-level targets with heterogeneous real and simulated sample sizes. We address this gap by calibrating quantile-based sim-to-real risk profiles under heterogeneous batching.

The remainder of the paper proceeds as follows. In Section~\ref{sec2} we present a motivating example, and formally state the problem. In Section~\ref{sec3} we present the main theoretical results, and in Section~\ref{sec4} we discuss an application of our methodology. We also address the tightness of our methodology to the actual quantile curve in Section~\ref{sec5}. In Section~\ref{sec6} we conclude by discussing future directions.   


\section{Framework and Motivating Example}\label{sec2}

We begin with a motivating use case 
of 
a survey-style LLM and then
formalize the quantile estimation problem. Although we introduce the
notation 
using
an LLM example, the framework itself is general. In
Appendix~\ref{extra_examples}, we illustrate how the same formulation 
applies 
to other settings, including discrete-event 
and other 
simulators with both categorical and continuous outputs.  

Suppose a media research company plans to survey customers on a particular
question with multinomial outcomes (e.g., `Agree', `Neutral',
`Disagree'). The company wants to estimate customer opinion before
committing resources to an expensive population study, and is considering
whether to directly use a digital twin based on pre-trained LLM or
develop its own fine-tuned LLM. 
  To make this decision, the company needs to evaluate 
  the distribution of the sim-to-real 
  gap across the range of questions that it is likely to encounter. 
  While pre-trained LLMs can appear to  provide coherent predictions, in 
  Appendix~\ref{apx:add_app} 
  we show examples where several pre-trained LLMs 
  underperform the baseline that randomly picked answers. 
The methods developed in this work assist such assessments.

Next, we formalize the setting and the associated theoretical challenges. There are two sources of uncertainty here: variation across scenarios and sampling noise within each scenario. First, there is a probability distribution $\scesp$ of scenarios
  (e.g. questions, operating conditions, etc.), 
  and
  the 
  dataset consists of $m$ \iid scenarios $\{\sce_j\}_{j=1}^m$
   from $\Psi$. For each scenario $\sce$, 
  the ground-truth 
  samples are drawn from the 
  conditional 
  real-world 
  distribution $Q^{rw}(\cdot\mid\sce)$, and 
  simulated 
  samples are drawn from the 
  conditional 
  simulator 
  distribution~$Q^{sim}(\cdot\mid\sce)$.  
  \footnote{We assume the LLM setting is fixed during calibration so that
    variation in simulation reflects scenario differences rather than
    encoding drift, hyperparameter choices, etc.; choosing a different
    setting effectively defines a different simulator.}
  In the LLM-survey example, the distributions $Q^{rw}(\cdot\mid\sce)$ and
  $Q^{sim}(\cdot\mid\sce)$ are obtained by marginalizing 
over human/synthetic profiles; see
Appendix~\ref{app:profile_level} for details.
  For each scenario $j\in[m]$, we
  observe a heterogeneous number of finite ground-truth samples: 
\begin{align*}
    y^{rw}_{j,i}\sim Q^{rw}(\cdot\mid\sce_j),
  \qquad &i=1,\ldots,n_j.
\end{align*}
For simplicity, we assume that we have the same number of synthetic samples $k$ from the
simulator:
\begin{align*}
y^{sim}_{j,\ell}\sim Q^{sim}(\cdot\mid\sce_j),
    \qquad &\ell=1,\ldots,k.
\end{align*}

Let $p(\sce)\in\Th$ denote a 
quantity 
of interest
corresponding to the 
real world
distribution
$Q^{rw}(\cdot\mid\sce)$, and let $q(\sce)\in\Th$ denote the
same quantity corresponding to the simulator distribution
$Q^{sim}(\cdot\mid\sce)$.
In the context of the media research company, $p(\sce)$ (resp. $q(\sce)$) denotes
the real-world (resp. simulator) probability distributions of the responses in scenario $\sce$. 
For simplicity, we write $p_\sce:=p(\sce)$ and $q_\sce:=q(\sce)$; for past
scenarios $\sce_j$, we write $p_j:=p(\sce_j)$ and $q_j:=q(\sce_j)$. 
The
samples $\{y^{rw}_{j,i}\}_{i=1}^{n_j}$ and $\{y^{sim}_{j,\ell}\}_{\ell=1}^{k}$
can be used to produce plug-in estimates $\hat p_j$ and $\hat q_j$ of $p_j$
and $q_j$, respectively. 

\begin{table}[ht]
\caption{Illustrative Example of multinomial distribution $\D$}\label{tbl:dataset}
\centering
\small
\setlength{\tabcolsep}{8pt}
\renewcommand{\arraystretch}{1.15}
\begin{tabularx}{\linewidth}{@{}C L N L@{}}
\toprule
\textbf{Scenario} & \textbf{True Distribution ($p_j$)} & $\boldsymbol{n_j}$ & \textbf{Real Output ($y^{\mathrm{rw}}$)} \\
\midrule
$1$ & $(0.8,\ 0.1,\ 0.1)$ & $10$ & $(1,\ 1,\ 2,\ 3,\ 1, 2,\ 1,\ 1,\ 1,\ 1)$ \\
$2$ & $(0.3,\ 0.4,\ 0.3)$ & $4$ & $(3,\ 1,\ 2,\ 2)$ \\
$\vdots$ & $\cdots$ & $\vdots$ & $\cdots$ \\
$m$ & $(0.05,\ 0.05,\ 0.9)$ & $8$ & $(1,\ 3,\ 3,\ 3, 3,\ 2,\ 2,\ 3)$ \\
\bottomrule
\end{tabularx}
\end{table}
 
The dataset is $\D = \big\{ (\sce_j,\hat p_j,\hat q_j,n_j, k)
\big\}_{j=1}^m$. For concreteness, Table~\ref{tbl:dataset} provides an example of $\D$ under our survey example. $p_j$ is a multinomial parameter and the human sample sizes $n_j$ may vary across scenarios, so the realized outcomes $\{y^{rw}_{j,i}\}_{i=1}^{n_j}$
differ in sample size. We record outcomes as category indices. For instance, when $j=2$ the true distribution is $p_2=(0.3,0.4,0.3)$ and a sample of size $n_2=4$ yields $\hat p_2=(0.25,0.5,0.25)$. The simulator analog $\hat q_j$ is constructed from $k$ simulator samples. 

\begin{remark}[Limited real-world samples]
In general, we require $n_j>1$ to avoid the single-observation identifiability issue in the absence of structural information linking scenarios. For example, consider independent observations $y^{\mathrm{rw}}_{j,1}\mid p_j\sim\mathrm{Bernoulli}(p_j)$. Taking $p_j\equiv1/2$ or drawing each $p_j$ independently from $\{0,1\}$ with equal probabilities produces identical single-observation laws, despite different latent parameter distributions. \footnote{In data-scarce settings, possible remedies include pooling comparable scenarios to assess group-level fidelity or imposing a parametric model linking $p_j$ to scenario covariates. These approaches entail a trade-off: pooling changes the target from individual scenarios to groups, while structural models introduce additional assumptions.}
\end{remark}

The discrepancy for a question $\sce$ between simulated and real output
distributions is defined as 
\[
  \Delta_\sce := L(p_\sce, \hat q_\sce),
\]
where $L : \Theta \times \Theta \to [0,\infty)$ is a user-chosen discrepancy function. We define the discrepancy in terms of $\hat q_\sce$ to reflect a fixed simulator query budget; when simulator sampling is effectively unlimited, one may replace $\hat q_\sce$ with $q_\sce$ to study the intrinsic discrepancy $L(p_\sce,q_\sce)$. Our method works with any function $L$ and allows practitioners to choose one that suits their application, e.g. Kullback--Leibler (KL) divergence for categorical outputs or a Wasserstein distance for empirical distributions. The resulting notion of fidelity depends on the chosen target and discrepancy $L$; Appendix~\ref{app:discrepancy_choice} provides discussion on these choices.

Now, we can formally state 
our task:

    \emph{Construct a 
  function $\hat V_m(\cdot,\cD) : (0,1) \to \mathbb{R}$ such that, for a new $\sce \sim \scesp$ and all $\tau \in (0,1)$,} 
\[ 
  \mP_{\sce \sim \scesp}\big( \Delta_\sce \le \hat V_m(\tau, \D) \big| \D \big) \geq \tau - \varepsilon_m,
\]
\emph{holds with high probability over the dataset $\D$ and $\lim_{m\rightarrow \infty}\varepsilon_m  = 0$. 
}

Note that $\hat V_m$ is indexed by both the prescribed quantile level $\tau$ and the 
dataset $\cD$. Establishing our desired result requires controlling two sources of uncertainty: (i) \emph{scenario uncertainty}, since we observe only finitely many scenarios $\{\sce_j\}_{j=1}^m \sim \scesp \rule[-1ex]{0pt}{0pt}$, and (ii) \emph{finite-sample uncertainty} for each scenario, since $p_j$ and $q_j$ are only observed through finite samples that produce $\hat p_j$ and $\hat q_j$. Our construction of $\hat V_m(\tau,\cD)$ explicitly accounts for both, and the main theorem will yield non-vacuous, finite-sample coverage guarantees. We summarize the notations introduced in Appendix~\ref{apx:notation}.


\section{Theoretical Result}\label{sec3}

In this section, we construct $\hat V_m,$ establish its theoretical guarantees, and conclude by discussing the benefits of distribution-level fidelity guarantees.

\subsection{Methodology}\label{subsec:method}
The dataset $\cD$ consists of $m$ scenarios with finite sample estimates $\hat p_j,\hat q_j$ of the latent targets $p_j,q_j$. Our procedure has two steps.
\begin{enumerate}
\item For each $j$, define the pseudo-discrepancy and its corresponding $\gamma_j$-confidence set $\cC_j(\hat p_j, \gamma_j)$
    \begin{equation}\label{eq:C-def}
       \hat{\Delta}_j := \sup_{u\in \Cj {\hat p_j, \gamma_j} } \Loss \big(u,\hat q_j\big), \text{ where } \  \mP\big(p_j \in \Cj{\hat p_j,\gamma_j}\mid \sce_j,n_j\big)\ge \gamma_j.
    \end{equation}
We will provide guidance on choosing the confidence level $\gamma_j$.
  \item 
    For $\tau \in (0,1)$, let 
    \begin{equation}
      \label{eq:Vhatm-def}
      \Vhatm{\tau} :=  \text{the $\tau$-quantile of}
      \{\hat\Delta_j\}_{j=1}^m.
    \end{equation}
  \end{enumerate}

We will use $\Vhatm{\tau}$ to construct a calibrated quantile curve for the discrepancy distribution, up to explicit finite-sample corrections.
Note that the key design choice in our method is the confidence set $\Cj \cdot$ that quantifies the uncertainty of $\hat{p}_j$. 
This uncertainty set is what links per-scenario sampling error to distribution-level calibration. Next, we show how to construct confidence sets using concentration bounds and compute the respective $\hat\Delta_j$ efficiently. We, again, adopt our multinomial example setting.
\begin{example}[Multinomial Confidence Set]\label{ex:multinom}   
  Suppose the outcomes are multinomial with parameter $p_j\in\Th=\mathcal{P}^d:=\{u\in[0,1]^d:\sum_{i=1}^d u_i=1\}$, i.e., $Y_{j,i}\mid \sce_j \sim \mathrm{Categorical}(p_j)$ and $\hat p_j$ are the empirical average vectors from $n_j$ samples.
  Let $\mathrm{KL}(\cdot \| \cdot)$ denote the KL divergence. Then the $\gamma_j$-confidence set for $p_j$ and the pseudo-discrepancy $\hat \Delta_j$ for loss $\Loss(u,\hat q_j)=\mathrm{KL}(u\|\hat q_j)$ is
  \eq
    \hat\Delta_j = \sup_{u\in\mathcal C_j(\hat p_j, \gamma_j)} \mathrm{KL} \big(u \| \hat q_j\big), \text{ where } \ \Cj{\hat p_j, \gamma_j} := \Big\{ u\in\mathcal{P}^d: \mathrm{KL}(\hat p_j\|u) \le \frac{d-1}{n_j}\ln(\frac{2(d-1)}{1-\gamma_j}) \Big\}.  \notag
      \label{eq:mult_ex}
  \en
  The confidence set results from a variant of Chernoff Bound (Lemma~\ref{lem:kl-tail-tight}). With strictly positive $\hat p_j$ and $\hat q_j$, the maximizer is attained on the boundary $\mathrm{KL}(\hat p_j\|u)=r_j$; computing this can be solved by standard constrained nonlinear solver.
\end{example}

More generally, $p_j$ may be any functional of the conditional output distribution, or the distribution itself; see Appendix~\ref{extra_output_examples} for additional examples. We now turn to the calibration guarantee for $\hat V_m$.


\subsection{Calibrated Quantile Curve Theory}\label{subsec:quantile_thm}
Our
theoretical guarantee relies on the following two assumptions.
\begin{assumption}[Independent data]\label{ass:indep} 
The scenarios $\psi_j\sim\Psi$ are i.i.d. Conditional on
$(\psi_1,\ldots,\psi_m)$, the real and simulated sample batches are independent
across scenarios, and within each scenario the real and simulated batches are
independent. A new scenario and its simulator samples are independent of
$\cD$. The human sample sizes $\{n_j\}_{j=1}^m$ are treated as fixed
deterministic design quantities.
\end{assumption}

\begin{assumption}[Regular Discrepancy] \label{ass:discr}
    The discrepancy $\Loss: \Th \times \Th\to[0,\infty)$ is jointly continuous on $\Th \times \Th$ and satisfies $\disc{u}{u}=0$ for all $u\in\Th$. 
\end{assumption}

Before stating our main result, we explain the intuition by considering the special case $\gamma_j \equiv \tfrac12$ and considering the asymptotic regime $m \to \infty$, in which our guarantee \eqref{eq:clean_affine_guarantee} implies the following: 
\begin{equation}\label{eq:informal_half}
    \mP_{\sce\sim\scesp}\Big( \Delta_\sce \le \hat
    V_m\Big(\tfrac{1+\tau}{2}\Big) \Big| \D \Big) \ge \tau, \text{ for all quantile level } \ \tau\in(0,1). 
\end{equation}
In words, the $\tau$-quantile of the discrepancy $\Delta_\sce$ is approximated by  the empirical pseudo-quantile curve at an \emph{inflated index} $(1 + \tau)/2$. This inflation induces an inherent conservativeness: for $\tau \approx 0$, the index $(1+\tau)/2$ is at least the median of the pseudo-discrepancy distribution,  leading to substantial inflation in the lower tail. To solve this, our methodology allows $\gamma_j$ to be different, and we obtain: 
\begin{equation}\label{eq:asmptotic_guarantee}
  \mP_{\sce\sim\scesp} \Big( \Delta_\sce \le \hat V_m \big( \bar \gamma_m \tau + (1-\bar \gamma_m) \big) \Big| \D\Big) \ge \tau,
\end{equation}
where $\bar{\gamma}_m = \tfrac{1}{m} \sum_{j=1}^m\gamma_j$. Now,  
to achieve $\tau$-level coverage we evaluate the empirical pseudo-quantile at an index that is affine in $\tau$ with slope $\bar\gamma_m$. One would ideally want $\bar{\gamma}_m = 1$ to control the index inflation. However, increasing $\gamma_j$  typically enlarges the set $\cC_j$, and thus, increases the values of the pseudo-discrepancies $\hat \Delta_j$ and shifts the curve $\hat V_m$ upward. Therefore, one needs to balance index adjustment with increased conservativeness of the pseudo-quantile curve.

This trade-off motivates choosing $\gamma_j$ adaptively as a function of the human sample size $n_j$. When $n_j$ is large, $p_j$ is typically more concentrated around $\hat{p}_j$; therefore, 
$\gamma_j$ can be set high without significantly enlarging $\cC_j$. Therefore, we set $\gamma_j=g(n_j)$ for an increasing function $g:\mathbb R_+\to(0,1)$. Two natural criteria for asymptotic tightness are $g(n)\to1$ and $r(n,g(n))\to0$, where $r(n,\gamma)$ denotes the confidence-set radius. Section~\ref{sec4} presents a specific choice and numerical evidence illustrating the conservativeness of fixed $\gamma$ and the benefits of adaptive coverage.\footnote{ Our average-coverage calibration relies on coverage levels being independent of the latent discrepancy ranks (scenario difficulty), as ensured by Assumption 3.2. For difficulty-dependent sample allocation, a possible remedy is to hold $\gamma_j\equiv\gamma$ constant and use additional samples to tighten $C_j$, provided conditional coverage and batch independence are maintained. Allowing coverage levels themselves to depend on difficulty requires additional assumptions or a different calibration, an extension beyond the scope of this paper.} We are now ready to state the main theorem rigorously.

\begin{theorem}\label{thm:clean_affine}
Suppose Assumptions~\ref{ass:indep}--\ref{ass:discr} hold and $\gamma_j \in (0,1)$ for each $j \in [m]$. Let $\cC_j(\hat{p}_j, \gamma_j)$ denote any $\gamma_j$-confidence set, $\hat{\Delta}_j$ denote the pseudo-discrepancy in~\eqref{eq:C-def}, and let $\hat{V}_m$ denote the quantile function of the pseudo-discrepancy  in~\eqref{eq:Vhatm-def}. Then, with probability at least $1-\delta$ over $\cD$, simultaneously for all
$\tau\in(0,1)$,
\begin{equation}\label{eq:clean_affine_guarantee}
  \mP_{\sce\sim\scesp} \Big( \Delta_\sce \le \hat V_m\big(\bar\gamma_m\tau+(1-\bar\gamma_m)\big) \Big| \cD \Big) \ge \tau-\varepsilon_m(\delta),
\end{equation}
where $\bar\gamma_m:=\frac1m\sum_{j=1}^m\gamma_j$ and $\varepsilon_m(\delta) := \sqrt{\frac{\log(6/\delta)}{2m}} +  \frac1m +  \frac{1}{\bar\gamma_m} \sqrt{\frac{2\log(3m/\delta)}{m}}.$ In particular,  $\varepsilon_m(\delta)\to0$ as $m\to\infty$.
\end{theorem}


Unlike standard conformal inference, which targets future observable
outcomes, Theorem~\ref{thm:clean_affine} provides uncertainty quantification
for latent scenario-level population quantities, such as conditional
output distributions or their functionals. Thus the guarantee is
distributional across scenarios, rather than a prediction guarantee for a
single future draw. We next provide a 
proof sketch.

\emph{Proof Sketch:}
Suppose we have access to the true $p_j$ for each scenario $j$. Then
$\{\Delta_j\}_{j=1}^m$ are \iid, and the Dvoretzky--Kiefer--Wolfowitz
(DKW) inequality yields a uniform deviation bound for the empirical CDF of discrepancy, which can be translated into upper bounds on its quantile function. In our setting, however, (i) $p_j$ is latent and (ii) each scenario is observed with a heterogeneous sample size $n_j$. 
Therefore, we aim to use empirical quantiles of the observable proxy $\{\hat\Delta_j\}$ to approximate the quantile function of the unobservable  $\{\Delta_j\}$. For simplicity, we fix $\gamma_j \equiv \gamma$. The key steps are:
\begin{enumerate}[(i)]
\item Let $x_\tau$ denote the population $\tau$-quantile of $\Delta_j$,
  and define the corresponding hardest scenarios $S_\tau:=\{j:\Delta_j\ge
  x_\tau\}$. By definition, $\tfrac{|S_\tau|}{m} \approx 1-\tau$. 
\item By construction, $\hat\Delta_j \ge \Delta_j$ holds with at least probability $\gamma$ for each $j$. Consequently, concentration implies that, with high probability, a $\gamma$-fraction of the indices in $S_\tau$, satisfy $\hat\Delta_j\ge\Delta_j\ge x_\tau$.
\item Hence, at least a $\gamma(1-\tau)$-fraction of the pseudo-discrepancies exceed $x_\tau$. In other words, the $\tau$-quantile of $\Delta_j$ is dominated by the $\big(1-\gamma(1-\tau)\big)$-quantile of $\hat\Delta_j$.
\end{enumerate}
Therefore, for any quantile level $\tau$, taking the $\big(\gamma\tau+(1-\gamma)\big)$-quantile of the pseudo-discrepancies yields an upper bound on the $\tau$-quantile of the oracle discrepancies. This is why the quantile index must shift: to upper bound the oracle $\tau$-quantile of $\Delta$, we need to use a more conservative pseudo-quantile level so that the resulting pseudo-threshold is high enough to cover the target quantile despite the miscoverage. This index adjustment is the first source of conservativeness; the second source is the DKW deviation that links empirical and population CDFs. 
Combining these two ingredients yields the theorem. We prove a sharper theorem in Appendix~\ref{pf:main}. The proof of Theorem~\ref{thm:clean_affine} then follows as a clean-index consequence of that sharper result and is provided in Appendix~\ref{apx:clean_affine_pf}.

\qed
 
\begin{remark}\label{rmk:k_budget}
  A subtle modeling choice is that Theorem~\ref{thm:clean_affine} actually characterizes the distribution of $\Delta^{(k)} = \Loss(p, \hat{q}^{(k)})$, where $\hat{q}^{(k)}$ is estimated using $k$ simulator samples, instead of the actual simulator parameter $q$. Conceptually, $\Delta^{(k)}$ is the correct target when one wants to characterize performance for a fixed query budget $k$. In contrast, $\Delta^* := \Loss(p, q)$ is more natural when simulator queries are cheap, and the goal also includes quantifying the simulator's inherent bias. For brevity, we ignored the $k$ superscript in our main text. A discussion to the latter case is presented in Corollary~\ref{cor:k_dependence}. 
\end{remark}

\begin{remark}
    The discrepancy in Theorem~\ref{thm:clean_affine} measures the sim-to-real gap. We can also extend this framework to the performance of simulators by comparing two simulators' discrepancies on the same scenarios, enabling paired statistical assessment of their relative fidelity. Complementing the marginal comparisons above, we present this extension in Appendix~\ref{apx_pairwise} and the experiments in Appendix~\ref{apx:wvb}.
\end{remark}

\subsection{Benefits of Distribution-Level Characterization}\label{subsec:dist_benefit}

We now turn to the benefits of obtaining a distribution-level
characterization of sim-to-real discrepancy. For illustrative purposes we set $\bar \gamma_m = \tfrac{1}{2}$ and define the index–adjusted (calibrated) curve $\hat V_m^{\mathrm{cal}}(\tau) := \hat V_m\Big(\tfrac{1+\tau}{2}\Big),$ for all $\tau\in[0,1],$ which matches the coverage guarantee in Equation~\ref{eq:informal_half} in asymptotic regime. With the calibrated quantile, for an unseen scenario $\sce$ we can derive a confidence set for $p_\sce$ given a simulated $\hat q_\sce$. For any confidence level $\bar\alpha\in(0,1)$, define 
\[
  S_{\bar\alpha} := \big\{u\in\Theta: \Loss(u,\hat q_\sce)\le \tau_{\bar\alpha} \big\}, \quad \tau_{\bar\alpha}:=\Vhatm{1-\tfrac{\bar\alpha}{2}}.
\]
By Theorem~\ref{thm:clean_affine}, $p_\sce\in S_{\bar\alpha}$ with probability at least $1-\bar\alpha$ up to an $o_m(1)$ remainder.

In addition, the calibrated quantile curve can be compressed to provide summary statistics and is therefore more informative than any single risk summary. Returning to the motivating example, a media research firm can use a calibrated AUC to assess average simulator bias, while a risk-averse product team may track a calibrated $\mathrm{CVaR}_{\alpha}$ to bound rare but consequential misreads of public sentiment. We can summarize overall average sim-to-real discrepancy and tail risk CVaR for any $\alpha \in(0,1)$ by,
\[
\mathrm{AUC}_{\mathrm{cal}} := \int_0^1 \hat V_m^{\mathrm{cal}}(\tau)\, d\tau, \qquad \mathrm{CVaR}^{\mathrm{cal}}_\alpha := \frac{1}{\alpha}\int_{1-\alpha}^{1} \hat V_m^{\mathrm{cal}}(u)\,du.
\]
In words, $\mathrm{AUC}_{\mathrm{cal}}$ aggregates the calibrated curve into a single average-bias summary, while $\mathrm{CVaR}^{\mathrm{cal}}_\alpha$ averages over the worst $\alpha$ fraction of scenarios under the same index-adjusted quantile curve.

\section{Application: LLM Fidelity Profiling}\label{sec4}

\subsection{Dataset and Methodology}

We evaluate our procedure on the \emph{WorldValueBench} dataset, curated
by~\cite{zhao2024worldvaluesbenchlargescalebenchmarkdataset} and
constructed from the World Values Survey~\citep{WVS_Wave7_2020}. The dataset comprises survey questions asked across 64 countries, covering 12 thematic categories (e.g., social values, security,
migration). Individual-level covariates are available for each respondent, including gender, age, migration status, education, marital status, and related demographics.  

\begin{figure}[ht!]
    \centering
    \includegraphics[width=0.6\linewidth]{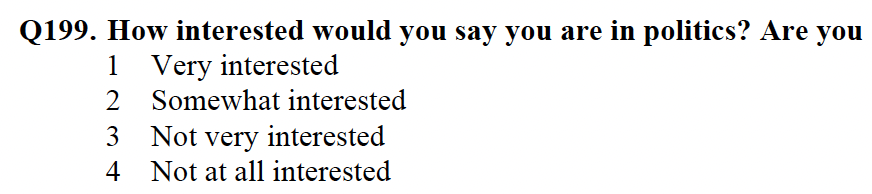}
    \caption{Example of World Value Question. Source: \cite{WVS_Wave7_2020}.}
    \label{fig:wvb_ex}
\end{figure}

\paragraph{Preprocessing.}
After data cleaning, we retain $235$ distinct questions and responses from
96{,}220 individuals. Each question offers a categorical set of possible
answers. Examples of the questions are provided in Figure~\ref{fig:wvb_ex}, and more in Appendix~\ref{apx:wvb}. To place heterogeneous 
answer sets on a common scale, we map each question's categories to the
interval $[-1,1]$, producing a bounded real-valued outcome for every scenario and respecting the natural ordering of response options when available (e.g., disagree–neutral–agree). Individual-level covariates are then used to construct synthetic profiles and prompts for simulators. Details on preprocessing and the dataset are provided in Appendix \ref{apx:wvb}. 

\paragraph{Methodology}
We generate synthetic responses and estimate $\{\hat q_j\}_{j=1}^{235}$
for four LLMs: \textsc{GPT-4o} (\texttt{gpt-4o}), \textsc{GPT-5\,mini}
(\texttt{gpt-5-mini}), \textsc{Llama 3.3 70B}
(\texttt{Llama-3.3-70B-Instruct-Turbo}), and \textsc{Qwen 3 235B}
(\texttt{Qwen3 235B A22B Thinking 2507 FP8}). For each LLM, we simulate $200$ responses per question to estimate $\hat q_j$. On the human side, we subsample $500$ respondents per question to compute $\{\hat p_j\}_{j=1}^{235}$ and construct the confidence sets $\mathcal C_j$; due to nonresponses, the effective sample sizes $n_j$ vary across questions, typically between $450$ and $500$. We discuss the relationship between our sampling assumptions and this particular application in Remark~\ref{rem:application_sampling}.

We include two baseline benchmarks: a \emph{uniform} baseline and a $K$-nearest-neighbor (KNN) predictor that uses OpenAI’s publicly available embedding API (\texttt{text-embedding-3-small}) to embed questions.\footnote{This API is no longer free to access, but it was at the time of writing.}. The uniform baseline samples uniformly from the available answers for each question. The KNN predictor estimates $\hat q_j$ from semantically similar questions: we use five-fold cross-validation over the 235 questions, and within each training split we tune the KNN hyperparameter and predict each held-out question's $\hat q_j$ using the human-side estimates $\{\hat p_\ell\}$ from its nearest training questions in embedding space. These benchmarks also illustrate the generality of our framework: as long as a method produces a simulator estimate $\hat q_j$, the same quantile-curve construction applies.

Since the outcome is bounded in $[-1,1]$, we construct per-scenario confidence sets $\mathcal C_j$ using Example~\ref{ex:bdd_cs}. We set $\gamma_j = 1-n_j^{-1/3}$, which is increasing in sample size $n_j$, \footnote{We adopt a generic form $\gamma_j=1-n_j^{-\beta}$ with $\beta = 1/3$, and conduct a sensitivity analysis on $\beta$ in Appendix~\ref{apx:gamma_choice}.} yielding an average coverage level of $\bar\gamma_m\approx 0.87$. We use
squared loss $\Loss(p,q)=(p-q)^2$. Applying the procedure in Section~\ref{sec3} (see~\eqref{eq:C-def}, \eqref{eq:Vhatm-def}) to each candidate simulator yields a calibrated fidelity profile, which we plot using the asymptotic index adjustment $\tau \mapsto \bar\gamma_m\tau+(1-\bar\gamma_m)$ in \eqref{eq:asmptotic_guarantee}. For completeness, we also apply the same workflow to different datasets, including Bernoulli setting (EEDI), multinomial setting (OpinionQA) and continuous temperature prediction setting (UMAR dataset from NEST); see Appendix~\ref{apx:add_app}. The full dataset, code, and reproduction materials are available in an anonymized repository for review at:
\url{https://github.com/yushiouwillylin/simulator_fidelity}.

\subsection{Fidelity Profiling}\label{subsec:fid_prof}

\begin{figure}[ht!]
    \centering
    \includegraphics[width=0.8\linewidth]{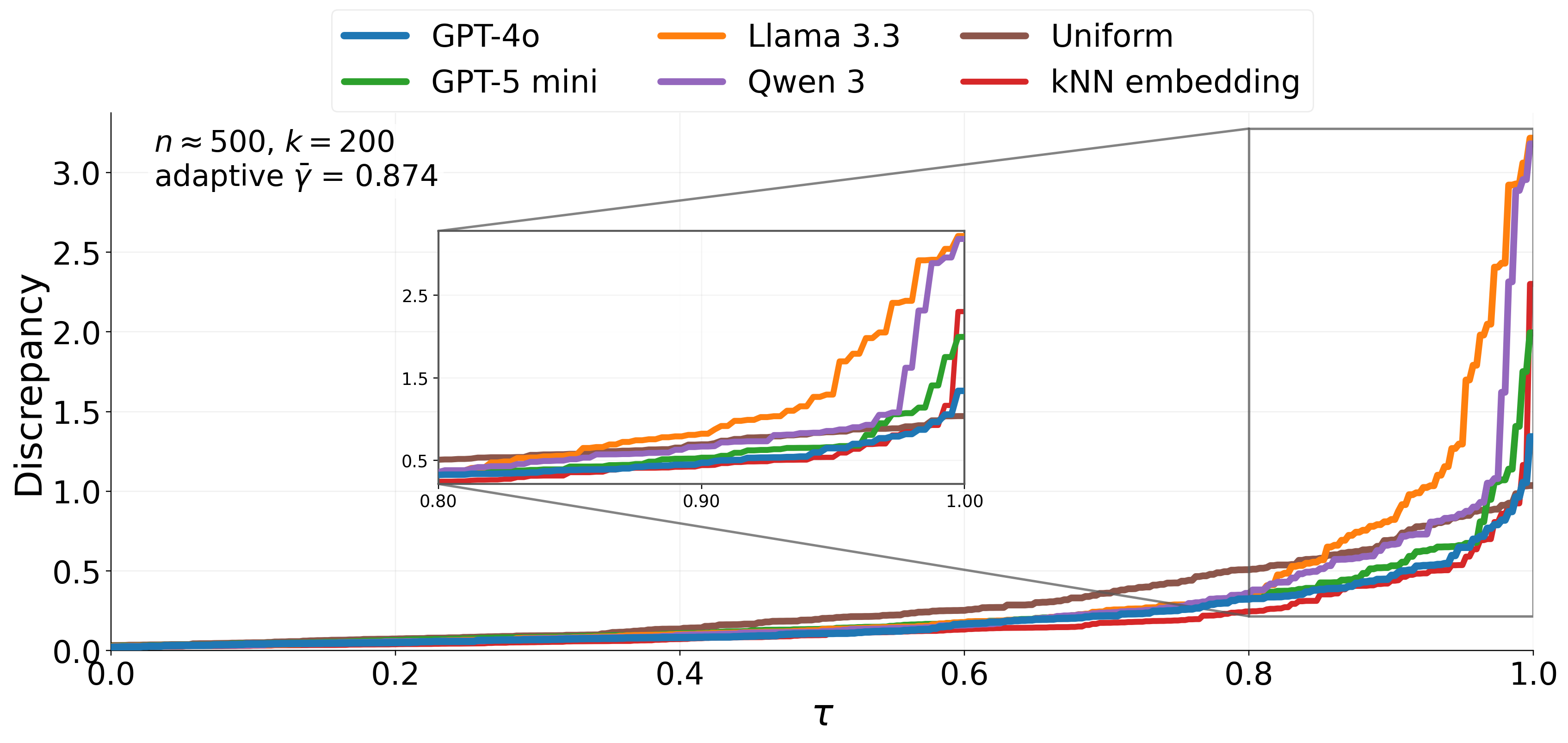}
    \caption{Calibrated discrepancy quantile curves $\hat V(\tau)$ across LLM simulators with zoom-in on upper 80-100 quantile. 
    }
    \label{fig:fidelity_profile_bddmean}
\end{figure}

Figure~\ref{fig:fidelity_profile_bddmean} plots the calibrated discrepancy quantile curves, which use the asymptotic index adjustment. For $m=235$ and $\delta=0.05$, the sharper finite-$m$ calibration requires an additional index adjustment and retains a coverage correction of approximately $0.105$; Appendix~\ref{apx:finite_m} presents the corresponding non-asymptotic curves, and discuss their informative range and the endpoint convention. Lower, flatter curves indicate smaller scenario-level sim-to-real gaps, while sharp elbows reveal rare but severe failures. Among the generative LLM simulators, \textsc{GPT-4o} has the most stable profile, with \textsc{GPT-5-mini} close behind; \textsc{Llama 3.3} and \textsc{Qwen-3} show substantially larger upper-tail discrepancies. All simulators improve over the uniform benchmark on over 80$\%$ of questions, but fail to capture tail scenarios. We conducted robustness checks under different human sample sizes $n_j$ and different query sample size $k$ in Section~\ref{sample_robustness}, ~\ref{apx:k_sensitivity}, and this pattern is stable.



The two baselines serve as diagnostics rather than competing black-box simulators. The uniform baseline is not intended to be a poor estimator, since for near-balanced questions a uniform response distribution can have a small sim-to-real gap. Its role is to provide a structure-free benchmark without using any semantic information about the question. The embedding-based KNN baseline plays a different diagnostic role: by predicting each held-out question from human response profiles of semantically nearby questions, it tests how much scenario-level response structure is captured by question embeddings alone. Its strong performance over most quantiles suggests that semantic similarity carries substantial information about human response distributions. Thus, the KNN result indicates semantic recoverability: question embeddings capture substantial scenario-level distributional structure, and our quantile profile assesses whether generative simulators recover it.

\section{Tightness Analysis of Calibrated Quantiles}\label{sec5}

A natural concern with our approach is the tightness 
of our calibrated quantile curve. 
Since our method uses off-the-shelf concentration bounds and worst-case pseudo-discrepancies, the calibrated curve may be conservative. We assess this empirically using the
large human sample in WorldValueBench. For each question, we treat the response
distribution estimated from all available respondents, over 90{,}000, as a proxy for the true $p_j$, which gives an oracle discrepancy $\Delta_j^\star := \Loss(p_j,\hat q_j)$ and hence an oracle quantile curve. We then fix the simulator budget at
$k=200$, subsample
$n_j\in\{50,500,1000,5000\}$ respondents per question, recompute the confidence
sets and pseudo-discrepancies, and compare the calibrated curves with the oracle
curve. The results plotted in Figure~\ref{fig:tightness} indicate that the calibrated envelope is substantially conservative when $n_j = 50$, but becomes tighter when $n_j \geq 500$. This behavior is consistent with our construction: larger $n_j$ yields concentrated estimates $\hat p_j$ and smaller confidence sets $\cC_j(\hat p_j,\gamma_j)$, which reduces the inflation induced by the supremum defining $\hat\Delta_j$. 

Figure~\ref{fig:tightness} uses the adaptive choice $\gamma_j=1-n_j^{-1/3}$. To isolate the role of this choice, as we discussed in Section~\ref{sec3}, we repeat the
experiment with fixed coverage $\gamma_j\equiv 1/2$; the results are shown in Figure~\ref{fig:tightness_fixed}. Unlike the adaptive case, the calibrated curves remain separated from the oracle curve even for large $n_j$. This persistent gap comes from the index adjustment. When
$\bar\gamma_m=1/2$, the asymptotic calibration evaluates the pseudo-quantile at $\hat V_m(0.5\tau+0.5)$. Thus, even near the lower tail $(\tau\approx 0)$, the
procedure uses the median pseudo-discrepancy, creating conservativeness that does not vanish as $n_j$ grows. This highlights the value of choosing adaptive coverage levels that changes with $n_j$.
 
\begin{figure}[htbp!]
    \centering

    \begin{subfigure}[t]{0.45\linewidth}
        \centering
        \includegraphics[width=\linewidth]{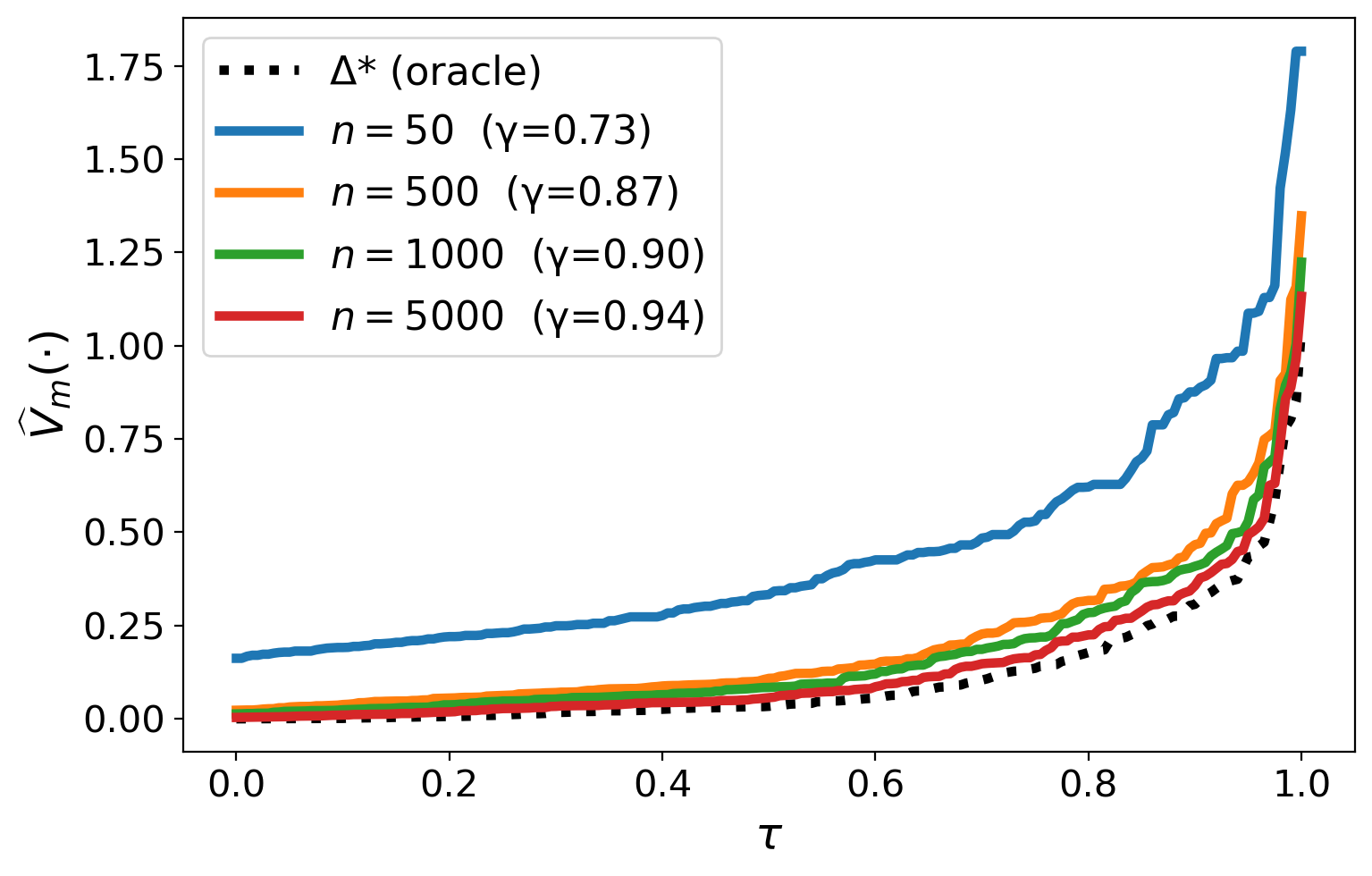}
        \caption{}
        \label{fig:tightness}
    \end{subfigure}
    \hfill
    \begin{subfigure}[t]{0.45\linewidth}
        \centering
        \includegraphics[width=\linewidth]{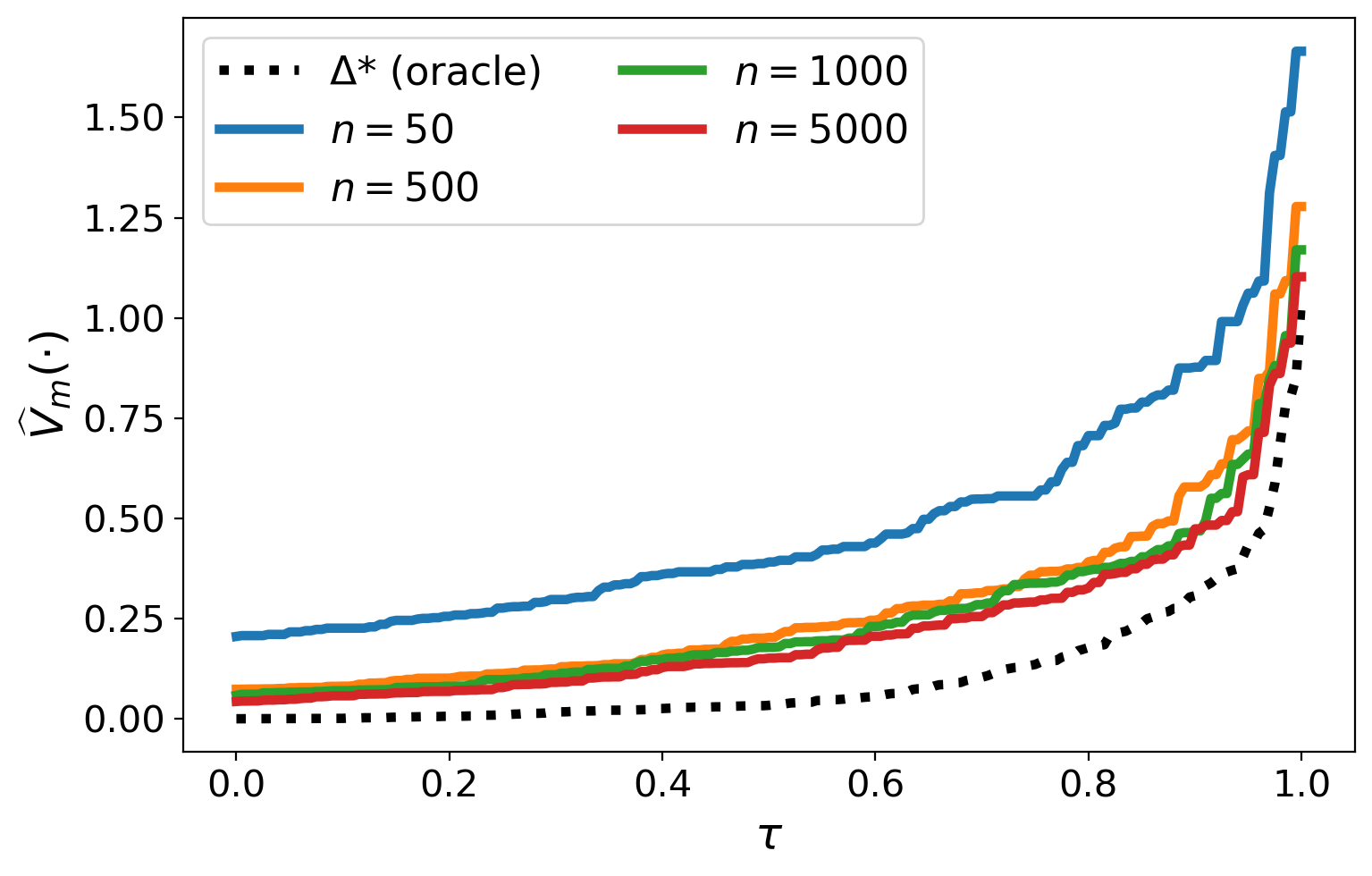}
        \caption{}
        \label{fig:tightness_fixed}
    \end{subfigure}

    \caption{Tightness analysis under GPT-4o. The dotted black curve is the oracle quantile curve for $\Delta^\star$; colored curves are computed with $n_j\in\{50,500,1000,5000\}$, and the adaptive average coverage levels are $\bar\gamma_m \approx 0.73, 0.87, 0.90,$ and $0.94$, respectively. Panel (b) uses fixed $\gamma_j=1/2$.
    }
    \label{fig:two_side_by_side}
\end{figure}

While WorldValueBench provides a large pool of respondents that facilitates an oracle tightness check, most applications do not enjoy such extensive ground-truth data. We next develop a complementary tightness assessment: a confidence band that sandwiches the oracle quantile curve.

\begin{theorem}[Informal Version]\label{thm:band_informal}
  Suppose the setup of Theorem~\ref{thm:clean_affine} and Assumptions~\ref{ass:indep}--\ref{ass:discr} hold. Let $\{\gamma_{L,j}, \gamma_{U,j} \}_{j=1}^m\subset(0,1]$ be coverage levels and define the lower and upper pseudo-discrepancies 
  \[
    \hat \Delta_j^- := \inf_{u\in \cC_j(\hat p_j,\gamma_{L,j})} \Loss(u, \hat q_j), \ \hat \Delta_j^+ := \sup_{u\in \cC_j(\hat p_j,\gamma_{U,j})} \Loss(u, \hat q_j),
  \]
  where $\cC_j(\hat p_j,\gamma)\subset\Theta$ 
  are $\gamma$-confidence sets. 
  Let $\hat V_m^-(\cdot)$ (resp. $\hat{V}_m^+(\cdot)$) denote the empirical quantile function of $\{\hat\Delta_j^-\}_{j=1}^m$ (resp. $\{\hat\Delta_j^+\}_{j=1}^m$). Let $V:[0,1]\to\mathbb R$ denote the quantile function of $\Delta(\psi)$. Then, as $m \to \infty$, for any fixed $\tau\in(0,1)$, 
  \begin{equation}\label{eq:band-quantile-asymp-gammaj}
    \hat V_m^-\big(\bar\gamma_{L,m}\tau\big) \lesssim V(\tau) \lesssim \hat V_m^+\big(\bar\gamma_{U,m} \tau+(1-\bar\gamma_{U,m}) \big),
  \end{equation}
  where 
  $\bar\gamma_{L,m} = \tfrac{1}{m} \sum_{j=1}^m \gamma_{L,j}$, and
  $\bar\gamma_{U,m} = \tfrac{1}{m}\sum_{j=1}^m \gamma_{U,j}$.
\end{theorem}

\begin{figure}[ht!]
    \centering
    \includegraphics[width=0.95\linewidth]{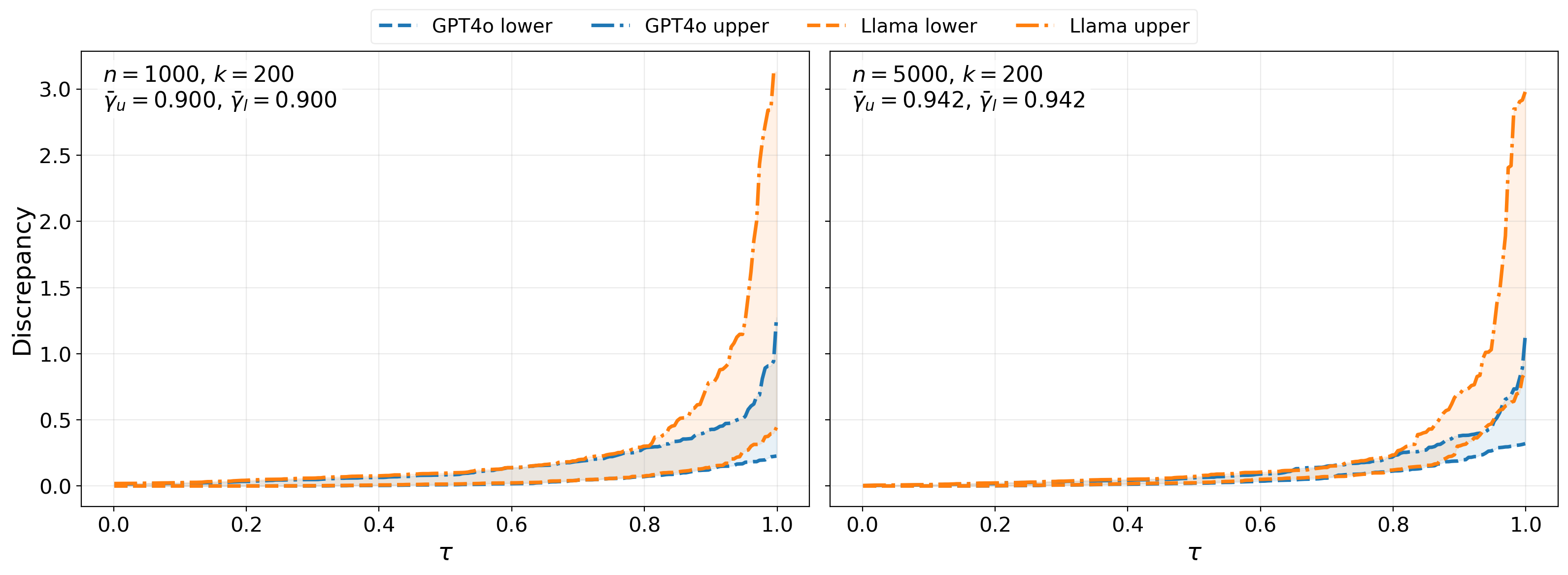}
    \caption{Confidence bands for GPT-4o and LLAMA 3.3. The left panel uses $n_j=1000$, and $\bar\gamma_L=\bar\gamma_U=0.9$; the right panel uses $n_j=5000$, and $\bar\gamma_L=\bar\gamma_U=0.94$, both with $k=200$. 
    }
   \label{fig:cb_ex}
\end{figure}

The formal finite-sample statement is given in Theorem~\ref{thm:band}. Together, Theorems~\ref{thm:clean_affine} and~\ref{thm:band} give complementary distributional fidelity guarantees: Theorem~\ref{thm:clean_affine} provides a one-sided calibrated upper envelope for the population quantile curve $V(\tau)$, while Theorem~\ref{thm:band} adds a lower envelope, yielding a confidence band for $V(\tau)$. Figure~\ref{fig:cb_ex} illustrates these bands for $n_j\in\{1000,5000\}$, comparing \textsc{GPT-4o} with \textsc{Llama 3.3}. Across both sample sizes, the \textsc{GPT-4o} band lies below the \textsc{Llama 3.3} band, especially in the upper tail. This suggests that \textsc{GPT-4o} has smaller scenario-level discrepancies and smaller extreme failures, and the conclusion is stable as $n_j$ increases. Importantly, the comparison is fully data-driven: the band construction does not require access to ground-truth discrepancies and can be carried out using only the observed heterogeneous batches. 


\section{Discussion}\label{sec6}

We develop a framework for profiling simulator fidelity through the quantile function of the sim-to-real discrepancy. Our approach makes no parametric assumptions about both the simulator/real systems, accounts for finite and heterogeneous sampling across scenarios, and yields finite-sample guarantees at any target quantile level. The resulting fidelity profile supports a range of summary statistics, from average discrepancy to tail-risk measures that capture rare but critical failures. Finally, we demonstrate the practical value of the framework on real survey data, evaluating multiple LLM-based simulators and illustrating distribution-level differences in their alignment with human populations.

Despite its broad applicability, several extensions remain. The current framework targets marginal
scenario-level population quantities, not respondent-level dependence across questions. It also assumes static i.i.d. scenarios and does not yet cover dynamic simulators with temporal dependence or distribution shift among scenarios. Moreover, generic concentration and DKW-type bounds can be conservative when the number of scenarios or per-scenario samples is small. Finally, on the computational side, the main bottleneck is simulator-side generation: once real and simulated batches are available, our methodology is computationally lightweight, but expensive simulators may make sample generation itself the limiting factor.

\clearpage

{
\bibliographystyle{ims}
\bibliography{ref}
}

\appendix

\section{Setup Discussions}

Our target is the quantile function of $\Delta_\psi=L(p_\psi,\widehat q_\psi)$ across scenarios $\psi\sim\Psi$, where $\widehat q_\psi$ is computed using a fixed simulator query budget. Although the sample estimates $\widehat p_j$ and $\widehat q_j$ are observable, the real-world parameter $p_j$, and hence the discrepancy $\Delta_j$, remains latent. For example, consider conditionally i.i.d.\ outcomes $y^{\mathrm{rw}}_{j,i}\mid\psi_j\sim\mathrm{Bernoulli}(p_j)$ and their sample mean $\widehat p_j=n_j^{-1}\sum_{i=1}^{n_j}y^{\mathrm{rw}}_{j,i}$. The outcome loss $(y^{\mathrm{rw}}_{j,i}-\widehat q_j)^2$ is observable, whereas the discrepancy $(p_j-\widehat q_j)^2$ is not. With fixed $n_j$ and conditionally independent real and simulator batches, replacing $p_j$ by $\widehat p_j$ yields
\[
\mathbb{E} \left[ (\widehat p_j-\widehat q_j)^2 \mid\psi_j,\widehat q_j,n_j \right] = (p_j-\widehat q_j)^2+\frac{p_j(1-p_j)}{n_j}.
\]
Thus, plug-in discrepancies incorporate sampling uncertainty that varies with $n_j$, and calibrating their quantiles does not directly calibrate the latent discrepancy distribution. Our two-level construction instead forms observable one-sided proxies using scenario-wise confidence sets and accounts for both within-scenario uncertainty and variation across the finitely many sampled scenarios. 

\subsection{Profile-Level Interpretation of the Survey Application}
\label{app:profile_level}

The main text works with the scenario-level conditional output distributions $Q^{rw}(\cdot\mid\sce)$ and $Q^{sim}(\cdot\mid\sce)$. In the survey application, these distributions can be obtained by marginalizing over respondent profiles. To include the respondent profile into the notation, let the real-world system be characterized by a latent profile $z\in\cZ$ distributed according to the population distribution $\cP$. For each pair $(\sce,z)$, the real-world system produces an outcome $Y^{rw}$ with conditional distribution
$Q^{rw}(\cdot\mid z,\sce)$. The scenario-level real output distribution used in
the main text is the population mixture
\[
    Q^{rw}(A\mid\sce) := \mE_{z\sim P} \left[ Q^{rw}(A\mid z,\sce) \right],
\]
for any measurable output set $A$.

In the categorical survey example, if
$Q^{rw}(\cdot\mid z,\sce)=\mathrm{Categorical}(\Pi^{rw}(\sce,z))$, then the
population response vector is
\[
    p_\sce = \mE_{z\sim P} \left[ \Pi^{rw}(\sce,z) \right].
\]
For scalar responses, $p_\sce$ may instead denote the corresponding marginal mean or another population-level functional of $Q^{rw}(\cdot\mid\sce)$. Similarly, the simulator draws a synthetic profile $z^{sim}\in\cZ^{sim}$ from a simulator population distribution $P^{sim}$ and generates an outcome according to $Q^{sim}(\cdot\mid z^{sim},\sce)$. The scenario-level simulator distribution is
\[
    Q^{sim}(A\mid\sce) := \mE_{z^{sim}\sim P^{sim}} \left[ Q^{sim}(A\mid z^{sim},\sce ) \right],
\]
and $ q_\sce = \mE_{z^{sim}\sim P^{sim}} \left[ \Pi^{sim}(\sce,z^{sim}) \right],$ with analogous definitions for means, functionals, or full distributions.

\subsection{Table of Notation}\label{apx:notation}

\begin{table}[ht!]
\centering
\small
\caption{Notation summary.}
\label{tab:notation}
\begin{tabular}{p{0.19\linewidth} p{0.74\linewidth}}
\toprule
Symbol & Meaning \\
\midrule
$\psi \sim \Psi$ & Scenario drawn from the scenario population. \\
$z \sim \cP$ & Real-world respondent drawn from the population distribution. \\
$p(\psi)$(or $p_\psi$) & Real-world population parameter for scenario $\psi$. \\
$q(\psi)$(or $q_\psi$) & Simulator-side population parameter for scenario $\psi$ under fixed simulator settings. \\
$\hat p_\sce, \hat q_\sce$ & Finite-sample estimators of $p_\sce$ and $q_\sce$. \\
$n_j$ & Human sample size for scenario $j$. \\
$k$ (or $k_j$) & Simulator sample size for scenario $j$. \\
$L(\cdot,\cdot)$ & Discrepancy function between real and simulated distribution. \\
$\Delta_\psi$ & True sim-to-real discrepancy for scenario $\psi$. \\
$\gamma_j$ & Per-scenario confidence-set coverage level. \\
$\bar \gamma_m$ & Average coverage level across $m$ scenarios. \\
$\cC_j(\hat p_j,\gamma_j)$ & $\gamma_j$-confidence set for the latent parameter $p_j$. \\
$\hat \Delta_j$ & Pseudo-discrepancy for scenario $j$, defined via a confidence set around $\hat p_j$. \\
$\hat V_m(\tau)$ & Pseudo-discrepancy quantile curve built from $m$ scenarios, indexed by quantile level $\tau$. \\
$\tau$ & Quantile-function argument, counted in the usual direction from lower to upper quantiles. \\
$\alpha$ & Upper-tail fraction / exceedance level used when discussing tail scenarios. \\
\bottomrule
\end{tabular}
\end{table}

\subsection{Choice and Interpretation of the Discrepancy}\label{app:discrepancy_choice}

Our framework does not claim to provide a universally optimal discrepancy. Instead, the target parameter and the discrepancy $L$ determine which aspects of simulator fidelity are assessed, while the calibrated quantile curve characterizes how the resulting discrepancy $\Delta_\psi=L(p_\psi,\widehat q_\psi)$ varies across scenarios. The framework is therefore agnostic to the simulator's internal model, but the meaning of fidelity remains relative to the chosen target, outcome representation, and discrepancy.

\paragraph{Choosing a discrepancy.}
The choice should reflect the application-relevant errors. Comparing means is appropriate when the average outcome is the quantity of interest, whereas comparing probability vectors or full distributions can reveal differences that averaging conceals. Table~\ref{tab:discrepancy_guidance} summarizes several choices. Each must be paired with confidence sets for the corresponding target and satisfy the stated conditions on the discrepancy.

\begin{table}[htbp]
\centering
\small
\renewcommand{\arraystretch}{1.15}
\caption{Discrepancy choices and their interpretation.}
\label{tab:discrepancy_guidance}
\resizebox{\linewidth}{!}{%
\begin{tabular}{@{}ll@{}}
\toprule
\textbf{Discrepancy} & \textbf{Primary sensitivity} \\
\midrule
Mean loss & Differences in the average outcome. \\[4pt]

Wasserstein distance $W_1$ & Distributional shifts measured by the distance over which probability mass is transported. \\[4pt]

Total variation & Redistribution of probability mass across categories. \\[4pt]

KL divergence & Outcomes to which the simulator assigns too little probability relative to the real distribution. \\[4pt]

MMD or energy distance & Distributional differences emphasized by the chosen kernel or outcome-space distance. \\
\bottomrule
\end{tabular}%
}
\end{table}

\section{Theorems Proof}\label{thm:proof}

\begin{theorem}\label{thm:general}
Suppose Assumptions~\ref{ass:indep}--\ref{ass:discr} hold and $\gamma_j
\in (0,1)$ for each $j \in [m]$. Let $\cC_j(\hat{p}_j, \gamma_j)$ denote
any $\gamma_j$-confidence set, $\hat{\Delta}_j$ denote the
pseudo-discrepancy (see~\eqref{eq:C-def}), and let $\hat{V}_m$ denote the
quantile function of the pseudo-discrepancy
(see~\eqref{eq:Vhatm-def}). For $\tau\in(0,1)$, let $[1-\tau]_m :=
m^{-1}\lceil m(1-\tau)\rceil$ denote the rounded empirical grid point
associated with the upper-tail level $1-\tau$, so that $[1-\tau]_m \in
[1-\tau,1-\tau+\frac1m]$, and let $[x]_+ := \max\{x,0\}$. Set 
\begin{equation}
  \label{eq:tau-eff-def}
  \tau_{\mathrm{eff}}(\tau) := 1-\Big[
  \big(\bar\gamma_m-\varepsilon_{\mathrm{al}}(\tau,\delta)\big)[1-\tau]_m
  - b_m(\delta)\sqrt{[1-\tau]_m} \Big]_+, 
\end{equation}
  where $\bar\gamma_m := \frac1m\sum_{j=1}^m\gamma_j$,  $\varepsilon_{\mathrm{al}}(\tau,\delta) =
\sqrt{\frac{\log(3m/\delta)}{2m [1-\tau]_m}}$, and $b_m(\delta) =
\sqrt{\frac{1}{2m}\log\frac{3m}{\delta}}$.\footnote{We use the extended
  endpoint convention that, when the expression inside 
$[\cdot]_+$ is zero, $\hat V_m(1)=+\infty$. Equivalently, under the usual
finite empirical convention $\hat V_m(1)=\max_j\hat\Delta_j$, the non-vacuous
finite-threshold guarantee below is asserted only for levels satisfying $
\big(\bar\gamma_m-\varepsilon_{\mathrm{al}}(\tau,\delta)\big)[1-\tau]_m -
b_m(\delta)\sqrt{[1-\tau]_m}>0.$} Then, with probability at least
$1-\delta$ over $\mathcal D$, simultaneously for all $\tau\in(0,1)$, 
\begin{equation}\label{eq:main_guarantee}
    \mP_{\sce\sim\scesp}\Big( \Delta_\sce \le \hat
    V_m\big(\tau_{\mathrm{eff}}(\tau)\big) \Big| \D \Big) \ge \tau -
    \varepsilon_m(\delta), 
\end{equation}
where $\varepsilon_m(\delta) := \sqrt{\frac{\log(6/\delta)}{2m}} +
\tfrac{1}{m}$.  In particular, for fixed $\delta$, $\varepsilon_m(\delta)\to0$ and $ \tau_{\mathrm{eff}}(\tau) \to \bar\gamma_m\tau+(1-\bar\gamma_m)$ as $m \to \infty$ uniformly over $\tau\in(0,1)$.   
\end{theorem}

\subsection{Proof for Theorem \ref{thm:general}}\label{pf:main}

We condition on $\cG_j:=\sigma(\psi_j,\hat q_j,n_j)$ and $\cG:=\sigma(\{\cG_j\}_{j=1}^m)$. By Assumption~\ref{ass:indep}, conditional on $\cG$ the additional randomness used to construct $\hat p_j$ (hence $\hat\Delta_j$ and the indicators below) is independent across $j$. For brevity, we abuse notation and write $p(\sce_j), q(\sce_j)$ as $p_j, q_j$, and similarly $\hat p, \hat q$. We also ignore the superscript $k$ to simplify the proof, and discussions on the dependence on $k$ can be found in Section~\ref{sec3}. Split $\delta = \delta_{\mathrm{DKW}} + \delta_{\mathrm{al}} + \delta_{\mathrm{cnt}}$ with $\delta_{\mathrm{DKW}} = \delta_{\mathrm{al}} = \delta_{\mathrm{cnt}} = \delta/3$, and union bound the three events defined below. Moreover, we prove the theorem in the form of probability of upper tail $\alpha := 1-\tau$ for simplicity, and will reparameterize at the end of the proof. For $\alpha\in(0,1)$, define the rounded upper-tail grid point $\alpha_+(\alpha):=\frac{\lceil m\alpha\rceil}{m}.$
Define
\[
  \varepsilon_{\mathrm{al}}(\alpha,\delta) := \sqrt{ \frac{\log(3m/\delta)} {2m\alpha_+(\alpha)} }, \quad b_m(\delta):= \sqrt{\frac{1}{2m}\log\frac{3m}{\delta}},
\]
and
\[
  \alpha_{\mathrm{eff}}(\alpha) := \left[ \big(\bar\gamma_m-\varepsilon_{\mathrm{al}}(\alpha,\delta)\big) \alpha_+(\alpha) - b_m(\delta)\sqrt{\alpha_+(\alpha)} \right]_+.
\]

For grid values $\alpha_r=r/m$, this reduces to
\[
  \alpha_{\mathrm{eff}}(\alpha_r) = \left[ \big(\bar\gamma_m-\varepsilon_{\mathrm{al}}(\alpha_r,\delta)\big)  \alpha_r - b_m(\delta)\sqrt{\alpha_r} \right]_+.
\]

We will show that these adjustment terms will coincide with the terms presented in the theorem statement after replacing $\alpha$ with $1-\tau$.

We first show an idealized case where we have access to the oracle $p_j$, and hence can obtain the actual sim-to-real discrepancy $\Delta_j$. Under this setting, we can define $F(t):=\mP_{\psi\sim\Psi}(\Delta(\psi)\le t)$ and 
\begin{align*}
    \bar F_m(t) := \frac1m\sum_{j=1}^m \mathbf 1\{\Delta_j\le t\}, \quad
    \bar V_m(u) := \inf\{t:\bar F_m(t)\ge u\}.
\end{align*}

Set $ \varepsilon_{\mathrm{DKW}}(\delta) := \sqrt{\frac{\log(6/\delta)}{2m}}.$ By the DKW inequality (Lemma~\ref{dkw_in}), the event
\[
  \cE_{\mathrm{DKW}} := \Big\{\sup_t|\bar F_m(t)-F(t)|\le \varepsilon_{\mathrm{DKW}}(\delta)\Big\}
\]
has probability at least $1-\delta_{\mathrm{DKW}}$. On $\cE_{\mathrm{DKW}}$, for every $u\in(0,1)$,
\[
  F\big(\bar V_m(u)\big) \ge \bar F_m\big(\bar V_m(u)\big)-\varepsilon_{\mathrm{DKW}}(\delta) \ge u-\varepsilon_{\mathrm{DKW}}(\delta).
\]
Equivalently,
\begin{equation}\label{eq:DKW_g}
  \mP_{\psi\sim\Psi} \big( \Delta(\psi) \le \bar V_m(u) \big| \mathcal D \big)  \ge u-\varepsilon_{\mathrm{DKW}}(\delta).
\end{equation}

Now in our dataset, we do not have access to the above oracle $\Delta_j$. Therefore, we control for $\Delta_j$ using the pseudo-discrepancy $\hat \Delta_j$. Define $Y_j:=\mathbf 1\{\hat\Delta_j\ge \Delta_j\}.$ By the coverage of $\cC_j(\hat p_j,\gamma_j)$ and the definition of $\hat\Delta_j$,
\[
  \mP(Y_j=1\mid \cG_j) \ge \gamma_j.
\]

For $r=1,\dots,m-1$, recall that $\alpha_r:=r/m$ and we define the oracle right-tail set
\[
  S_r:=\{j\in[m]:\Delta_j\ge \bar V_m(1-\alpha_r)\}, \quad s_r:=|S_r|.
\]
Let $\Delta_{(1)}\le\cdots\le\Delta_{(m)}$ denote the order statistics of
$\{\Delta_j\}_{j=1}^m$. Since
\[
  \bar V_m(1-\alpha_r) = \Delta_{(m-r)},
\]
at most $m-r-1$ observations are strictly below
$\bar V_m(1-\alpha_r)$. Hence $s_r \ge r+1.$
 
We then link the oracle set to the pseudo-discrepancies by defining
\[
  T_r := \sum_{j\in S_r}Y_j, \quad
  U_r := \bigl | \{j:\hat\Delta_j\ge \bar V_m(1-\alpha_r)\}\bigr|.
\]
By construction, $U_r\ge T_r$, and $U_r$ will give us the number of $\hat \Delta$ required to achieve the tail count of the empirical oracle quantile function, hence we will try to create a lower bound for $U_r$ via $T_r$. 

Since $\gamma_j=g(n_j)$ and the sample sizes $n_j$ are fixed design quantities under Assumption~\ref{ass:indep}, the vector $(\gamma_1,\ldots,\gamma_m)$ is fixed independently of the exchangeable vector $(\Delta_1,\ldots,\Delta_m)$. Hence, conditional on $s_r$ and on the realized multiset $\{\gamma_j\}_{j=1}^m$, the set $S_r$ is a uniformly random subset of $[m]$ of size $s_r$. Therefore, the collection $\{\gamma_j:j\in S_r\}$ is a simple random sample without replacement from the finite population $\{\gamma_j\}_{j=1}^m$. By Hoeffding--Serfling \cite{hoeffding_serfling}, for each $r$ and each $t>0$,
\[
  \mP\Big(\frac{1}{s_r} \sum_{j\in S_r}\gamma_j \le \bar\gamma_m-t \Big) \le \exp(-2s_rt^2).
\]
Set
\[
 t_r := \sqrt{\frac{\log(m/\delta_{\mathrm{al}})}{2s_r}} \le \sqrt{\frac{\log(3m/\delta)}{2r}} = \sqrt{\frac{\log(3m/\delta)}{2m\alpha_r}} = \varepsilon_{\mathrm{al}}(\alpha_r,\delta).
\]
A union bound over $r=1,\dots,m-1$ gives an event $\cE_{\mathrm{al}}$ with $\mP(\cE_{\mathrm{al}})\ge 1-\delta_{\mathrm{al}}$ on which, for all $r$,
\begin{equation}\label{eq:align_g}
  \sum_{j\in S_r}\gamma_j\ \ge\ s_r\big(\bar\gamma_m-\varepsilon_{\mathrm{al}}(\alpha_r,\delta)\big).
\end{equation}
Conditional on $\cG$, the indicators $\{Y_j\}_{j=1}^m$ are independent and bounded in $[0,1]$, so Lemma \ref{lem:add-chernoff} implies that for any $t>0$,
\[
  \mP\Big(T_r\le \mE[T_r\mid\cG]-t \Big| \cG\Big) \le \exp\Big(-\frac{2t^2}{s_r}\Big).
\]
Since $\mE[T_r\mid\cG]=\sum_{j\in S_r}\mE[Y_j\mid\cG]\ge \sum_{j\in S_r}\gamma_j$, choose
\[
  c_m:=\sqrt{\frac12 \log \frac{m}{\delta_{\mathrm{cnt}}}} = \sqrt{\frac12\log\frac{3m}{\delta}}, \quad   t_r:=c_m\sqrt{s_r}.
\]
Then $\mP(T_r\ge \sum_{j\in S_r}\gamma_j-c_m\sqrt{s_r}\mid\cG)\ge 1-\delta_{\mathrm{cnt}}/m$. A union bound over $r=1,\dots,m-1$ yields an event $\cE_{\mathrm{cnt}}$ with $\mP(\cE_{\mathrm{cnt}})\ge 1-\delta_{\mathrm{cnt}}$ on which, for all $r=1,\dots,m-1$,
\begin{equation}\label{eq:count_g}
  T_r \ge \sum_{j\in S_r}\gamma_j - c_m\sqrt{s_r}.
\end{equation}

Combining \eqref{eq:align_g} and \eqref{eq:count_g} gives, for all
$r=1,\dots,m-1$,
\[
  T_r
  \ge
  s_r\big(\bar\gamma_m-\varepsilon_{\mathrm{al}}(\alpha_r,\delta)\big)
  -
  c_m\sqrt{s_r}
\]
on $\cE_{\mathrm{al}}\cap\cE_{\mathrm{cnt}}$. We now have a computable lower bound for $T_r$ up to $s_r$. Let
\[
  A_r := \bar\gamma_m-\varepsilon_{\mathrm{al}}(\alpha_r,\delta), \qquad f_r(s):=A_rs-c_m\sqrt{s},
\]
and
\[
  \beta_r
  :=
  \alpha_{\mathrm{eff}}(\alpha_r)
  =
  \left[A_r\alpha_r-b_m(\delta)\sqrt{\alpha_r}\right]_+.
\]
Then
\[
  U_r\ge T_r\ge f_r(s_r).
\]

If $\beta_r=0$, then $1-\beta_r=1$, and the desired comparison $\hat V_m(1-\beta_r)\ge \bar V_m(1-\alpha_r)$ is immediate under the extended endpoint convention $\hat V_m(1)=+\infty$. Hence assume $\beta_r>0$. Since $r=m\alpha_r$ and $b_m(\delta)=c_m/\sqrt m$,
\[
  f_r(r) = rA_r-c_m\sqrt r= m\beta_r > 0.
\]
Thus $A_r\sqrt r>c_m$, and for all $s\ge r$,
\[
  f_r'(s) = A_r-\frac{c_m}{2\sqrt{s}} \ge A_r-\frac{c_m}{2\sqrt r} > 0.
\]
Since $s_r\ge r+1$, by the positive derivative, we obtain
\[
  U_r \ge T_r \ge f_r(s_r) > f_r(r) = m\beta_r.
\]
As $U_r$ is integer,
\[
  U_r\ge \lfloor m\beta_r\rfloor+1.
\]
Therefore at least $\lfloor m\beta_r\rfloor+1$ of the
$\hat\Delta_j$'s are at least $\bar V_m(1-\alpha_r)$. By the
order-statistic thresholding lemma~\ref{lem:os-threshold}, with
$N=\lfloor m\beta_r\rfloor+1$,
\[
  \hat\Delta_{(m-N+1)}
  =
  \hat\Delta_{(m-\lfloor m\beta_r\rfloor)}
  \ge
  \bar V_m(1-\alpha_r).
\]
Since
\[
  \left\lceil m(1-\beta_r)\right\rceil
  =
  m-\lfloor m\beta_r\rfloor,
\]
this is exactly
\begin{equation}\label{eq:quant_link_g}
  \hat V_m\big(1-\alpha_{\mathrm{eff}}(\alpha_r)\big)
  \ge
  \bar V_m(1-\alpha_r),
  \qquad r=1,\dots,m-1.
\end{equation}

Finally, work on
$\cE_{\mathrm{DKW}}\cap\cE_{\mathrm{al}}\cap\cE_{\mathrm{cnt}}$, which has probability at least $1-\delta$ by the union bound. Fix $\alpha\in(0,1)$ and set $\alpha_+ := \frac{\lceil m\alpha\rceil}{m}.$ If $\alpha_+=1$, then $ 1-\alpha-\varepsilon_{\mathrm{DKW}}(\delta)-\frac1m<0,$ so the desired lower bound is trivial. Hence assume $\alpha_+<1$ and let $r=\lceil m\alpha\rceil\in\{1,\dots,m-1\}$. By the definition of $\alpha_{\mathrm{eff}}(\alpha)$ through the rounded grid point, $\alpha_{\mathrm{eff}}(\alpha) = \alpha_{\mathrm{eff}}(\alpha_+).$ Applying \eqref{eq:quant_link_g} at this $r$ gives
\[
  \hat V_m\big(1-\alpha_{\mathrm{eff}}(\alpha)\big) =  \hat V_m\big(1-\alpha_{\mathrm{eff}}(\alpha_+)\big) \ge \bar V_m(1-\alpha_+).
\]
Applying \eqref{eq:DKW_g} with $u=1-\alpha_+$ gives
\begin{align*}
  \mP_{\psi\sim\Psi}
  \big(
    \Delta(\psi)\le \bar V_m(1-\alpha_+)
    \,\big|\,
    \mathcal D
  \big)
  &\ge
  1-\alpha_+-\varepsilon_{\mathrm{DKW}}(\delta) \\
  &\ge
  1-\alpha-\varepsilon_{\mathrm{DKW}}(\delta)-\frac1m.
\end{align*}
By monotonicity in the threshold,
\[
  \mP_{\psi\sim\Psi}
  \big(
    \Delta(\psi)
    \le
    \hat V_m\big(1-\alpha_{\mathrm{eff}}(\alpha)\big)
    \,\big|\,
    \mathcal D
  \big)
  \ge
  1-\alpha-\varepsilon_{\mathrm{DKW}}(\delta)-\frac1m.
\]
Define
\[
  \varepsilon_m(\delta)
  :=
  \varepsilon_{\mathrm{DKW}}(\delta)+\frac1m
  =
  \sqrt{\frac{\log(6/\delta)}{2m}}+\frac1m.
\]
This proves the $\alpha$-form of the guarantee.

It remains only to reparameterize from the upper-tail level $\alpha$ to the coverage level $\tau$. Set $\alpha=1-\tau.$ Then $ \alpha_+(\alpha) = \frac{\lceil m(1-\tau)\rceil}{m} = [1-\tau]_m.$ Moreover,
\[
  \alpha_{\mathrm{eff}}(1-\tau)
  =
  \left[
    \big(\bar\gamma_m-\varepsilon_{\mathrm{al}}(\tau,\delta)\big)
    [1-\tau]_m
    -
    b_m(\delta)\sqrt{[1-\tau]_m}
  \right]_+, \text{ where } \
  \varepsilon_{\mathrm{al}}(\tau,\delta)
  =
  \sqrt{
    \frac{\log(3m/\delta)}
         {2m[1-\tau]_m}
  }.
\]
By definition, $ \tau_{\mathrm{eff}}(\tau) := 1-\alpha_{\mathrm{eff}}(1-\tau).$ Therefore,
\[
  \hat V_m\big(1-\alpha_{\mathrm{eff}}(1-\tau)\big)
  =
  \hat V_m\big(\tau_{\mathrm{eff}}(\tau)\big),
\]
and the $\alpha$-form guarantee becomes
\[
  \mP_{\sce\sim\scesp} \Big( \Delta_\sce \le \hat V_m\big(\tau_{\mathrm{eff}}(\tau)\big) \Big|
    \mathcal D \Big) \ge \tau-\varepsilon_m(\delta),
\]
where
\[
  \varepsilon_m(\delta)
  =
  \sqrt{\frac{\log(6/\delta)}{2m}}+\frac1m.
\]
This is exactly \eqref{eq:main_guarantee}.

\qed

\subsection{Proof for Theorem~\ref{thm:clean_affine}}\label{apx:clean_affine_pf}

\begin{proof}
By Theorem~\ref{thm:general}, with probability at least $1-\delta$ over
$\cD$, simultaneously for all $\alpha\in(0,1)$,
\[
  \mP_{\sce\sim\scesp} \Big( \Delta_\sce \le \hat V_m\big(1-\alpha_{\mathrm{eff}}(\alpha)\big) \Big|    \cD \Big) \ge 1-\alpha-\varepsilon_m(\delta),
\]
where $ \varepsilon_m(\delta) = \sqrt{\frac{\log(6/\delta)}{2m}}+\frac1m.$ Our goal is to get a simpler form as in Theorem~\ref{thm:clean_affine}, and we observe that from the definition of
$\alpha_{\mathrm{eff}}(\cdot)$,
\[
  \alpha_{\mathrm{eff}}(\alpha') \ge \bar\gamma_m\alpha'_+ - \sqrt{\frac{2\alpha'_+}{m}\log\frac{3m}{\delta}} \ge  \bar\gamma_m\alpha' -  \sqrt{\frac{2}{m}\log\frac{3m}{\delta}}.
\]
Thus, to guarantee $ \alpha_{\mathrm{eff}}(\alpha')\ge \bar\gamma_m\alpha,$ it is enough to choose $\alpha'=\alpha+\rho_m(\delta)$ with
\[
  \bar\gamma_m\rho_m(\delta) = \sqrt{\frac{2}{m}\log\frac{3m}{\delta}} \implies \rho_m(\delta) = \frac{1}{\bar\gamma_m} \sqrt{\frac{2\log(3m/\delta)}{m}}.
\]

Now, if $\alpha+\rho_m(\delta)\ge 1$, then $ 1 - \alpha -\varepsilon_m(\delta)-\rho_m(\delta)\le 0,$ so the desired lower bound is trivial. Hence assume $\alpha+\rho_m(\delta)<1$,
and by applying Theorem~\ref{thm:general} at $\alpha'$ gives
\[
  \mP_{\sce\sim\scesp} \Big( \Delta_\sce \le \hat V_m\big(1-\alpha_{\mathrm{eff}}(\alpha')\big) \Big| \cD \Big) \ge 1-\alpha'-\varepsilon_m(\delta).
\]

The clean pseudo-quantile level should have upper-tail mass
$\bar\gamma_m\alpha$, since only a $\bar\gamma_m$-fraction of the oracle
upper tail is certified by the pseudo-discrepancies. Thus the corresponding
pseudo-quantile index is $1-\bar\gamma_m\alpha$. By the choice of
$\rho_m(\delta)$ above, we have
$\alpha_{\mathrm{eff}}(\alpha')\ge \bar\gamma_m\alpha$, and hence
$1-\alpha_{\mathrm{eff}}(\alpha')\le 1-\bar\gamma_m\alpha$. Since
$\hat V_m(\cdot)$ is nondecreasing,
\[
  \hat V_m\big(1-\alpha_{\mathrm{eff}}(\alpha')\big)
  \le
  \hat V_m\big(1-\bar\gamma_m\alpha\big).
\]
Therefore, by monotonicity of the event in the threshold,
\[
  \mP_{\sce\sim\scesp} \Big(
    \Delta_\sce \le \hat V_m\big(1-\bar\gamma_m\alpha\big)
    \Big| \cD \Big)
  \ge
  1-\alpha'-\varepsilon_m(\delta)
  =
  1-\alpha-\rho_m(\delta)-\varepsilon_m(\delta).
\]

Finally, return to the coverage notation $\tau=1-\alpha$. Then
$1-\bar\gamma_m\alpha=\bar\gamma_m\tau+(1-\bar\gamma_m)$, and the last display
becomes
\[
  \mP_{\sce\sim\scesp} \Big(
    \Delta_\sce
    \le
    \hat V_m\big(\bar\gamma_m\tau+(1-\bar\gamma_m)\big)
    \Big| \cD \Big)
  \ge
  \tau-\widetilde\varepsilon_m(\delta),
\]
where
\[
  \widetilde\varepsilon_m(\delta)
  :=
  \varepsilon_m(\delta)+\rho_m(\delta)
  =
  \sqrt{\frac{\log(6/\delta)}{2m}}
  +
  \frac1m
  +
  \frac{1}{\bar\gamma_m}
  \sqrt{\frac{2\log(3m/\delta)}{m}}.
\]
Since Theorem~\ref{thm:general} holds simultaneously for all $\alpha\in(0,1)$, the preceding argument holds simultaneously for all $\tau\in(0,1)$. This proves Theorem~\ref{thm:clean_affine}.
\end{proof}

\begin{corollary}\label{thm:general_gamma_fixed}
Adopt the setting of Theorem~\ref{thm:general} and fix $\gamma\in(0,1)$. Set $\gamma_j=\gamma$ for all $j\in[m]$. Then, for every $\eta\in(0,1)$, with probability at least $1-\eta$ over $\D$, simultaneously for all $\alpha\in(0,1)$,
\begin{equation}\label{eq:main_gamma}
    \mP_{\sce\sim\scesp}\Big( \Delta_\sce^{(k)} \le \Vhatm{1-\gamma\alpha} \Big| \D \Big) \ge 1-\alpha-\frac{\varepsilon_\gamma(\alpha,m,\eta)}{\sqrt m},
\end{equation}
where
\[
    \varepsilon_\gamma(\alpha,m,\eta) := \sqrt{\frac12\log\frac6\eta} + \frac1{\sqrt m} + \frac1\gamma\sqrt{2\log\frac{3m}{\eta}}.
\]
The inner probability averages over a new scenario and its $k$ simulator samples. For fixed $\gamma$ and $\eta$, $\varepsilon_\gamma(\alpha,m,\eta)=O(\sqrt{\log m})$ uniformly over $\alpha\in(0,1)$.
\end{corollary}
\begin{proof}
Apply Theorem~\ref{thm:clean_affine} with
$\gamma_j=\gamma$ for all $j\in[m]$, $\delta=\eta$,
and $\tau=1-\alpha$. Then $\bar\gamma_m=\gamma$, so
the quantile index becomes
\[
    \bar\gamma_m\tau+(1-\bar\gamma_m)
    =
    \gamma(1-\alpha)+(1-\gamma)
    =
    1-\gamma\alpha.
\]
The corresponding coverage remainder is
\[
    \sqrt{\frac{\log(6/\eta)}{2m}}
    +\frac1m
    +\frac1\gamma
       \sqrt{\frac{2\log(3m/\eta)}m}
    =
    \frac{\varepsilon_\gamma(\alpha,m,\eta)}{\sqrt m}.
\]
Since the theorem holds simultaneously for all
$\tau\in(0,1)$, the substitution $\tau=1-\alpha$
gives \eqref{eq:main_gamma} simultaneously for all
$\alpha\in(0,1)$.
\end{proof}

\subsection{Corollary for simulator budget:}\label{pf:budget}

\begin{corollary}\label{cor:k_dependence}
Suppose Assumptions~\ref{ass:indep} and~\ref{ass:discr}
hold, allowing fixed, possibly heterogeneous simulator
sample sizes $k_j\in\mathbb N$. Define the intrinsic
sim-to-real discrepancy and its pseudo-discrepancy by
\[
    \Delta_j^\star := \Loss(p_j,q_j),
    \qquad
    \hat\Delta_j^{\star,(k)}
    :=
    \sup_{\substack{
        u\in C_j^p(\hat p_j)\\
        v\in C_j^q(\hat q_j)
    }}
    \Loss(u,v),
\]
where $C_j^p(\hat p_j),C_j^q(\hat q_j)\subset\Theta$
are nonempty compact confidence sets constructed from
the corresponding real and simulated sample batches
of scenario $j$, and satisfy
\[
    \mP\Big(
        p_j\in C_j^p(\hat p_j),\
        q_j\in C_j^q(\hat q_j)
        \Big| \sce_j,n_j,k_j
    \Big)
    \ge \frac12.
\]
Let $\hat V_m^\star(\alpha)$ denote the empirical
$\alpha$-quantile of
$\{\hat\Delta_j^{\star,(k)}\}_{j=1}^m$.

Then, for every $\eta\in(0,1)$, with probability at least
$1-\eta$ over $\D$, simultaneously for all
$\alpha\in(0,1)$,
\begin{equation}\label{eq:main-full-pq}
    \mP_{\sce\sim\scesp}\Big(
        \Delta_\sce^\star
        \le
        \hat V_m^\star\big(1-\tfrac{\alpha}{2}\big)
        \Big| \D
    \Big)
    \ge
    1-\alpha-\frac{\varepsilon(\alpha,m,\eta)}{\sqrt m},
\end{equation}
where
\[
    \varepsilon(\alpha,m,\eta)
    :=
    \varepsilon_{1/2}(\alpha,m,\eta)
    =
    \sqrt{\frac12\log\frac6\eta}
    +\frac1{\sqrt m}
    +2\sqrt{2\log\frac{3m}{\eta}},
\]
using the remainder defined in
Corollary~\ref{thm:general_gamma_fixed}.
For fixed $\eta$,
$\varepsilon(\alpha,m,\eta)/\sqrt m
=O(\sqrt{\log m/m})$ uniformly over $\alpha\in(0,1)$.
\end{corollary}

\begin{proof}
Define
\[
    \cG_j^\star:=\sigma(\psi_j,n_j,k_j),
    \qquad
    \cG^\star:=
    \sigma(\cG_1^\star,\ldots,\cG_m^\star),
    \qquad
    Y_j^\star:=
    \mathbf 1\{
        \hat\Delta_j^{\star,(k)}\ge\Delta_j^\star
    \}.
\]
Both the real and simulated sampling batches remain
random conditional on $\cG^\star$.
On the joint confidence event
\[
    \{p_j\in C_j^p(\hat p_j),\
      q_j\in C_j^q(\hat q_j)\},
\]
the definition of the supremum gives
$\hat\Delta_j^{\star,(k)}\ge\Delta_j^\star$.
Consequently, the assumed joint coverage and
independence across scenarios imply
\[
    \mP(Y_j^\star=1\mid\cG^\star)\ge\frac12,
\]
and the indicators $Y_1^\star,\ldots,Y_m^\star$
are conditionally independent.

Moreover,
$\Delta_j^\star=\Loss(p(\psi_j),q(\psi_j))$
are i.i.d., because the scenarios are i.i.d.
The fixed sample sizes $n_j$ and $k_j$ do not enter
these oracle discrepancies. Thus the oracle DKW bound
and conditional counting argument in the proof of
Theorem~\ref{thm:general} apply to
$(\Delta_j^\star,\hat\Delta_j^{\star,(k)})$
with $\gamma_j\equiv\tfrac12$.
The affine-index reduction in the proof of
Theorem~\ref{thm:clean_affine} then gives
\eqref{eq:main-full-pq}, simultaneously for all
$\alpha\in(0,1)$, with
$\varepsilon(\alpha,m,\eta)
=\varepsilon_{1/2}(\alpha,m,\eta)$.
\end{proof}
\begin{theorem}\label{thm:band}
  Suppose the setup of Theorem~\ref{thm:general} and Assumptions~\ref{ass:indep}--\ref{ass:discr} hold. Let $\{\gamma_{L,j}, \gamma_{U,j} \}_{j=1}^m\subset(0,1]$ be coverage levels and define the lower and upper pseudo-discrepancies
  \[
    \hat \Delta_j^- := \inf_{u\in \cC_j(\hat p_j,\gamma_{L,j})} \Loss(u, \hat q_j), \quad \hat \Delta_j^+ := \sup_{u\in \cC_j(\hat p_j,\gamma_{U,j})} \Loss(u, \hat q_j),
  \]
  where $\cC_j(\hat p_j,\gamma)\subset\Theta$ are data-driven compact confidence sets satisfying the
  conditional coverage guarantees
  \begin{align*}
      \mP \big(p_j\in \cC_j(\hat p_j,\gamma_{L,j}) \mid \sce_j,n_j\big)&\ge \gamma_{L,j} \\
      \mP\big(p_j\in \cC_j(\hat p_j,\gamma_{U,j}) 
    \mid \sce_j,n_j\big)&\ge \gamma_{U,j}.
  \end{align*}
  Let $\hat V_m^-(\cdot)$ denote the empirical quantile function of $\{\hat\Delta_j^-\}_{j=1}^m$, and $\hat V_m^+(\cdot)$ the empirical quantile function of $\{\hat\Delta_j^+\}_{j=1}^m$. Define the empirical mean coverage levels $\bar\gamma_{L,m}:=\frac{1}{m}\sum_{j=1}^m \gamma_{L,j}, \bar\gamma_{U,m}:=\frac{1}{m}\sum_{j=1}^m \gamma_{U,j}.$ For $\alpha\in(0,1)$ and $\delta\in(0,1)$ define the discretized level $ \tilde\alpha(\alpha):=\frac{\lceil m\alpha\rceil}{m}.$ Set $\varepsilon_{\mathrm{al}}^{\mathrm{bd}}(\alpha,\delta) := \sqrt{ \frac{\log(6m/\delta)} {2m\tilde\alpha(\alpha)} }, \text{ and } b_m^{\mathrm{bd}}(\delta) := \sqrt{\frac{1}{2m}\log\frac{6m}{\delta}}.$ Define the effective indices
\begin{align*}
  \alpha_{\mathrm{eff}}^{L}(\alpha)
  &:=
  \Big[
    \big(\bar\gamma_{L,m}
    -
    \varepsilon_{\mathrm{al}}^{\mathrm{bd}}(\alpha,\delta)\big)
    \tilde\alpha(\alpha)
    -
    b_m^{\mathrm{bd}}(\delta)\sqrt{\tilde\alpha(\alpha)}
  \Big]_+, \\
  \alpha_{\mathrm{eff}}^{U}(\alpha)
  &:= \Big[ \big(\bar\gamma_{U,m} -   \varepsilon_{\mathrm{al}}^{\mathrm{bd}}(\alpha,\delta)\big)
    \tilde\alpha(\alpha)
    -
    b_m^{\mathrm{bd}}(\delta)\sqrt{\tilde\alpha(\alpha)}
  \Big]_+.
\end{align*}
  Then, with probability at least $1-\delta$ over $\cD$, \emph{simultaneously for all $\alpha\in(0,1)$},
  \begin{align*}
  \mP_{\psi\sim\Psi}\Big(\Delta(\psi) \ge \hat V_m^-\big(\alpha_{\mathrm{eff}}^{L}(\alpha)\big) \Big| \cD\Big) & \ge 1-\alpha-\varepsilon_m^{\mathrm{bd}}(\delta),\\
    \mP_{\psi\sim\Psi}\Big(\Delta(\psi) \le \hat V_m^+\big(1-\alpha_{\mathrm{eff}}^{U}(\alpha)\big) \Big| \cD\Big) & \ge 1-\alpha-\varepsilon_m^{\mathrm{bd}}(\delta),
  \end{align*}
   where $ \varepsilon_m^{\mathrm{bd}}(\delta) := \sqrt{\frac{\log(12/\delta)}{2m}}+\frac{1}{m}.$ We use the endpoint conventions $\hat V_m^-(0)=-\infty$ and $\hat V_m^+(1)=+\infty$ whenever the corresponding effective index degenerates to zero. Let $V:[0,1]\to\mathbb R$ denote the quantile function of $\Delta(\psi)$.  In particular, for any fixed $\tau\in(0,1)$, as $m \to \infty$, the band admits the heuristic form
  \begin{equation}\label{eq:band-quantile-asymp-gammaj}
    \hat V_m^-\big(\bar\gamma_{L,m}\tau\big) \lesssim V(\tau) \lesssim \hat V_m^+\big(\bar\gamma_{U,m} \tau+(1-\bar\gamma_{U,m}) \big)
  \end{equation}
\end{theorem}

\subsection{Proof for Theorem~\ref{thm:band}:}\label{pf:band}

We prove the lower-band inequality with failure probability $\delta/2$. The upper-band inequality follows by applying the $\alpha$-form of Theorem~\ref{thm:general} with failure probability $\delta/2$ to the upper pseudo-discrepancies $\{\hat\Delta_j^+\}_{j=1}^m$ and adaptive coverage levels $\{\gamma_{U,j}\}_{j=1}^m$. Combining the two one-sided events by a union bound then gives the claimed probability $1-\delta$.

Recall that we work conditionally on $ \cG_j:=\sigma(\psi_j,\hat q_j,n_j), \ \cG:=\sigma(\{\cG_j\}_{j=1}^m).$ The true per-scenario discrepancy is $ \Delta_j:=\Loss(p_j,\hat q_j),$ for $j=1,\ldots,m.$ Define the lower pseudo-discrepancies and validity indicators by
\[
  \hat\Delta_j^- := \inf_{u\in \cC_j(\hat p_j,\gamma_{L,j})} \Loss(u,\hat q_j), \quad  Y_j^-:=\mathbf 1\{\hat\Delta_j^-\le \Delta_j\}.
\]
By Assumption~\ref{ass:discr} and compactness of $\cC_j(\hat p_j,\gamma_{L,j})$, the infimum is attained. Moreover, on the event $\{p_j\in \cC_j(\hat p_j,\gamma_{L,j})\}$, we have $\hat\Delta_j^-\le \Delta_j$. Therefore, by conditional coverage,
\[
  \mP(Y_j^-=1\mid \cG_j)\ge \gamma_{L,j}.
\]
Conditional on $\cG$, the indicators $\{Y_j^-\}_{j=1}^m$ are independent.

Set $ \delta_{\mathrm{DKW}}=\delta_{\mathrm{al}}=\delta_{\mathrm{cnt}}=\delta/6$ and  $c_m:=\sqrt{\frac12\log\frac{6m}{\delta}}.$ For $r=1,\ldots,m$, let $\alpha_r:=r/m$ and define the oracle lower-tail set
\[
  S_r^-:=\{j\in[m]:\Delta_j\le \bar V_m(\alpha_r)\}, \quad s_r^-:=|S_r^-|.
\]
By the definition of empirical quantiles, $s_r^-\ge r$. Define the counts
\[
  T_r^-:=\sum_{j\in S_r^-}Y_j^-, \quad U_r^-:=\big|\{j:\hat\Delta_j^-\le \bar V_m(\alpha_r)\}\big|.
\]
By construction, $U_r^-\ge T_r^-$. Since the coverage levels $\{\gamma_{L,j}\}_{j=1}^m$ are fixed design quantities and $(\Delta_1,\ldots,\Delta_m)$ is exchangeable, conditional on $s_r^-$ and on the realized multiset $\{\gamma_{L,j}\}_{j=1}^m$, the set $S_r^-$ is a uniformly random subset of $[m]$ of size $s_r^-$. Hence $\{\gamma_{L,j}:j\in S_r^-\}$ is a simple random sample without replacement from the finite population $\{\gamma_{L,j}\}_{j=1}^m$. By Hoeffding--Serfling and a union bound over $r=1,\ldots,m$, there is an event $\cE_{\mathrm{al}}^-$ with probability at least $1-\delta_{\mathrm{al}}$ on which, for all $r=1,\ldots,m$,
\[
  \sum_{j\in S_r^-}\gamma_{L,j} \ge s_r^- \Big( \bar\gamma_{L,m} - \varepsilon_{\mathrm{al}}^{\mathrm{bd}}(\alpha_r,\delta) \Big).
\]

Next, conditional on $\cG$, the indicators $\{Y_j^-\}_{j=1}^m$ are independent and bounded. By the additive Chernoff bound and a union bound over $r=1,\ldots,m$, there is an event $\cE_{\mathrm{cnt}}^-$ with probability at least $1-\delta_{\mathrm{cnt}}$ on which, for all $r=1,\ldots,m$,
\[
  T_r^- \ge \sum_{j\in S_r^-}\gamma_{L,j} - c_m\sqrt{s_r^-}.
\]

On $\cE_{\mathrm{al}}^-\cap\cE_{\mathrm{cnt}}^-$, define
\[
  A_r^-:= \bar\gamma_{L,m} - \varepsilon_{\mathrm{al}}^{\mathrm{bd}}(\alpha_r,\delta), \quad f_r^-(s):=A_r^-s-c_m\sqrt{s},
\]
and $ \beta_r^-:= \alpha_{\mathrm{eff}}^L(\alpha_r) = \left[ A_r^-\alpha_r - b_m^{\mathrm{bd}}(\delta)\sqrt{\alpha_r} \right]_+.$ Then
\[
  U_r^-\ge T_r^-\ge f_r^-(s_r^-).
\]

If $\beta_r^-=0$, the desired lower quantile comparison is immediate under the endpoint convention $\hat V_m^-(0)=-\infty$. Suppose $\beta_r^->0$. Since $r=m\alpha_r$ and $b_m^{\mathrm{bd}}(\delta)=c_m/\sqrt m$,
\[
  f_r^-(r) = rA_r^- - c_m\sqrt r = m\beta_r^- > 0.
\]
Thus $A_r^-\sqrt r>c_m$, and for all $s\ge r$,
\[
  (f_r^-)'(s)
  =
  A_r^-
  -
  \frac{c_m}{2\sqrt{s}}
  \ge
  A_r^-
  -
  \frac{c_m}{2\sqrt r}
  >
  0.
\]
Since $s_r^-\ge r$, we obtain
\[
  U_r^-
  \ge
  T_r^-
  \ge
  f_r^-(s_r^-)
  \ge
  f_r^-(r)
  =
  m\beta_r^-.
\]
As $U_r^-$ is integer, $ U_r^-\ge \lceil m\beta_r^-\rceil.$ Therefore at least $\lceil m\beta_r^-\rceil$ of the lower pseudo-discrepancies are no larger than $\bar V_m(\alpha_r)$. Hence, by the order-statistic thresholding lemma,
\begin{equation}\label{eq:lower_quant_link_adaptive}
  \hat V_m^-\big(\alpha_{\mathrm{eff}}^L(\alpha_r)\big) \le \bar V_m(\alpha_r), \quad r=1,\ldots,m.
\end{equation}

It remains to transfer this empirical oracle comparison to the population distribution. For the lower-band argument, use the strict CDFs
\[
  F_\Delta^{<}(t):=\mP_{\psi\sim\Psi}(\Delta(\psi)<t),
  \qquad
  \bar F_m^{<}(t):=\frac1m\sum_{j=1}^m\mathbf 1\{\Delta_j<t\}.
\]
The DKW inequality applies equally to these strict CDFs. Thus, with probability
at least $1-\delta_{\mathrm{DKW}}$,
\[
  \cE_{\mathrm{DKW}}^{<} := \left\{ \sup_t \big| \bar F_m^{<}(t)-F_\Delta^{<}(t) \big| \le \sqrt{\frac{\log(12/\delta)}{2m}} \right\}.
\]

Work on
$\cE_{\mathrm{DKW}}^{<}\cap\cE_{\mathrm{al}}^-\cap\cE_{\mathrm{cnt}}^-$. Fix any $\alpha\in(0,1)$ and set $ \alpha_+:=\frac{\lceil m\alpha\rceil}{m}.$ Let $r=\lceil m\alpha\rceil$. By the definition of $\alpha_{\mathrm{eff}}^L(\alpha)$ through the rounded grid point,
\[
  \alpha_{\mathrm{eff}}^L(\alpha)
  =
  \alpha_{\mathrm{eff}}^L(\alpha_+).
\]
Applying \eqref{eq:lower_quant_link_adaptive} at this $r$ gives
\[
  \hat V_m^-\big(\alpha_{\mathrm{eff}}^L(\alpha)\big)
  \le
  \bar V_m(\alpha_+).
\]
Therefore, by monotonicity and the strict-DKW event,
\[
  F_\Delta^{<}
  \Big(
    \hat V_m^-\big(\alpha_{\mathrm{eff}}^L(\alpha)\big)
  \Big)
  \le
  F_\Delta^{<}\big(\bar V_m(\alpha_+)\big)
  \le
  \bar F_m^{<}\big(\bar V_m(\alpha_+)\big)
  +
  \sqrt{\frac{\log(12/\delta)}{2m}}.
\]
Since $\alpha_+$ is a grid point and
$\bar V_m(\alpha_+)=\Delta_{(\lceil m\alpha_+\rceil)}$, at most
$\lceil m\alpha_+\rceil-1$ oracle discrepancies are strictly below
$\bar V_m(\alpha_+)$. Hence
\[
  \bar F_m^{<}\big(\bar V_m(\alpha_+)\big)
  \le
  \alpha_+.
\]
Thus
\[
  F_\Delta^{<}
  \Big(
    \hat V_m^-\big(\alpha_{\mathrm{eff}}^L(\alpha)\big)
  \Big)
  \le
  \alpha
  +
  \frac1m
  +
  \sqrt{\frac{\log(12/\delta)}{2m}}.
\]
Equivalently,
\[
  \mP_{\psi\sim\Psi}
  \Big(
    \Delta(\psi)
    \ge
    \hat V_m^-\big(\alpha_{\mathrm{eff}}^L(\alpha)\big)
    \,\Big|\,
    \cD
  \Big)
  \ge
  1-\alpha-\varepsilon_m^{\mathrm{bd}}(\delta).
\]
This proves the adaptive lower-band inequality simultaneously for all
$\alpha\in(0,1)$ with failure probability at most $\delta/2$.

The adaptive upper-band inequality follows from Theorem~\ref{thm:general},
applied with failure probability $\delta/2$ to the upper pseudo-discrepancies
$\{\hat\Delta_j^+\}_{j=1}^m$ and coverage levels
$\{\gamma_{U,j}\}_{j=1}^m$. Under this substitution,
\[
  \bar\gamma_m \mapsto \bar\gamma_{U,m},
  \qquad
  \varepsilon_{\mathrm{al}}(\alpha,\delta/2)
  =
  \varepsilon_{\mathrm{al}}^{\mathrm{bd}}(\alpha,\delta),
  \qquad
  b_m(\delta/2)
  =
  b_m^{\mathrm{bd}}(\delta),
\]
and the DKW remainder becomes
\[
  \sqrt{\frac{\log(12/\delta)}{2m}}+\frac1m
  =
  \varepsilon_m^{\mathrm{bd}}(\delta).
\]
Therefore, with probability at least $1-\delta/2$,
simultaneously for all $\alpha\in(0,1)$,
\[
  \mP_{\psi\sim\Psi}
  \Big(
    \Delta(\psi)
    \le
    \hat V_m^+\big(1-\alpha_{\mathrm{eff}}^U(\alpha)\big)
    \,\Big|\,
    \cD
  \Big)
  \ge
  1-\alpha-\varepsilon_m^{\mathrm{bd}}(\delta).
\]
Combining the lower and upper one-sided events by a union bound gives the
claimed simultaneous band with probability at least $1-\delta$.
\qed

\subsection{Technical Lemma}
\begin{lemma}(Dvoretzky-Kiefer-Wolfowitz (DKW) Inequality via \cite{massart1990tight})\label{dkw_in}

    Let \( X_1, X_2, \dots, X_n \) be i.i.d. real-valued random variables with cumulative distribution function (CDF) \( F^* \), and let \( \hat{F}_n \) be the empirical distribution function defined by
\[
\hat{F}_n(x) := \frac{1}{n} \sum_{i=1}^n \mathbf{1}\{X_i \leq x\}.
\]

Then, for any \( \varepsilon > 0 \),
\[
\mathbb{P} \left( \sup_{x \in \mathbb{R}} \left| \hat{F}_n(x) - F^*(x) \right| > \varepsilon \right)
\leq 2 e^{-2n \varepsilon^2}.
\]

Equivalently, for any confidence level \( \delta \in (0,1) \), with probability at least \( 1 - \delta \),
\[
\sup_{x \in \mathbb{R}} \left| \hat{F}_n(x) - F^*(x) \right| \leq \sqrt{ \frac{1}{2n} \log\left( \frac{2}{\delta} \right) }.
\]
\end{lemma}

\begin{lemma}\label{lem:add-chernoff}[Hoeffding (additive) via \cite{Hoeffding01031963}]

Let $Z_1,\dots,Z_n\in[0,1]$ be independent, $T=\sum_{i=1}^n Z_i$, and $\mu=\mathbb{E}[T]$.
For any $t\in[0,\mu]$,
\[
\mP \big(T \le \mu - t\big) \le 
\exp\left(-\frac{2t^2}{\sum_{i=1}^n (1-0)^2}\right)
=\exp \left(-\,\frac{2t^2}{n}\right).
\]
\end{lemma}

\begin{lemma}\label{lem:CH_exp}[Chernoff--Hoeffding for one-parameter exponential family]

Let \(X_1,\dots,X_n\) be i.i.d. with density (or mass) in the one-parameter canonical exponential family
\[
p_\theta(x) \;=\; \exp\{\theta T(x) - A(\theta)\}\,h(x),
\qquad \theta\in\Theta,
\]
where \(T(x)\) is the sufficient statistic, \(A(\theta)\) is the log-partition function (convex, differentiable on \(\Theta\)), and the mean map is \(\mu(\theta):= \mathbb{E}_{\theta}[T(X)] = A'(\theta)\). Assume \(\Theta\) is an open interval and all quantities below are finite.

Define the empirical mean of the sufficient statistic
\[
\bar T_n \;=\; \frac{1}{n}\sum_{i=1}^n T(X_i).
\]
For each \(t\) in the range of \(\bar T_n\) let \(\theta_t\) be the (unique) canonical parameter satisfying \(\mu(\theta_t)=t\). 
Then for any \(\varepsilon>0\),
\begin{align*}
    &\mP \bigl( D(p_{\theta_{\bar T_n}}\| p_\theta) > \varepsilon \bigr) \le 2 e^{-n\varepsilon}
\end{align*}
\end{lemma}

\begin{proof}
Write $\mu=A'(\theta)$. For $t\in A'(\Theta)$, let
$\theta_t$ satisfy $A'(\theta_t)=t$ and define
\[
I_\theta(t)
:=D(p_{\theta_t}\|p_\theta)
=(\theta_t-\theta)t-A(\theta_t)+A(\theta).
\]
The last equality follows directly from the exponential-family
density.

For any $\lambda$ such that $\theta+\lambda\in\Theta$,
independence gives
\[
\mathbb E_\theta\!\left[e^{\lambda n\bar T_n}\right]
=
\exp\!\left(n\{A(\theta+\lambda)-A(\theta)\}\right).
\]

Suppose first that $t>\mu$. Since the mean map is strictly
increasing, $\lambda_t:=\theta_t-\theta>0$. Markov's inequality
therefore gives
\begin{align*}
\mathbb P_\theta(\bar T_n\ge t)
&\le
e^{-n\lambda_t t}
\mathbb E_\theta\!\left[e^{n\lambda_t\bar T_n}\right]\\
&=
\exp\!\left(
-n\{\lambda_t t-A(\theta+\lambda_t)+A(\theta)\}
\right)\\
&=e^{-nI_\theta(t)}.
\end{align*}

If $t<\mu$, then $\lambda_t<0$, so the corresponding lower-tail
event satisfies
\begin{align*}
\mathbb P_\theta(\bar T_n\le t)
&=
\mathbb P_\theta\!\left(
e^{n\lambda_t\bar T_n}\ge e^{n\lambda_t t}
\right)\\
&\le
e^{-n\lambda_t t}
\mathbb E_\theta\!\left[e^{n\lambda_t\bar T_n}\right]\\
&=e^{-nI_\theta(t)}.
\end{align*}
At $t=\mu$, both bounds reduce to the trivial bound $1$.

Moreover,
\[
I_\theta(\mu)=0,
\qquad
I_\theta'(t)=\theta_t-\theta.
\]
Thus $I_\theta$ is decreasing to the left of $\mu$ and increasing
to the right of $\mu$.

Fix $\varepsilon>0$. On each side of $\mu$ where
$I_\theta(t)>\varepsilon$ is nonempty, continuity and monotonicity
give a threshold $t_-$ or $t_+$ satisfying
\[
t_-<\mu<t_+,
\qquad
I_\theta(t_-)=I_\theta(t_+)=\varepsilon,
\]
with a threshold omitted if its corresponding tail is empty.
Consequently,
\[
\{I_\theta(\bar T_n)>\varepsilon\}
\subseteq
\{\bar T_n\le t_-\}\cup\{\bar T_n\ge t_+\},
\]
again omitting any empty tail. Applying the appropriate
one-sided bounds and the union bound yields
\[
\mathbb P_\theta\!\left(
D(p_{\theta_{\bar T_n}}\|p_\theta)>\varepsilon
\right)
\le e^{-n\varepsilon}+e^{-n\varepsilon}
=2e^{-n\varepsilon}.
\]
\end{proof}

\begin{lemma} (Chernoff-Hoeffding Inequality)\label{CH-bin}

Let \( X_1, \ldots, X_n \sim \text{Ber}(\tilde{p}) \) be i.i.d. Bernoulli random variables with unknown mean \( \tilde{p} \), and define the empirical mean as
\[
\bar{X}_n = \frac{1}{n} \sum_{i=1}^n X_i.
\]

Then for any \( \varepsilon > 0 \),
\[
\mathbb{P} \left( D\left( \bar{X}_n \,\|\, \tilde{p} \right) > \varepsilon \right) \leq 2 e^{-n \varepsilon},
\]
where \( D(p \,\|\, q) \) is the Kullback-Leibler divergence between Bernoulli distributions with parameters \( p \) and \( q \).    
\end{lemma}

\begin{proof}
First suppose $\tilde p\in(0,1)$. The Bernoulli family has
canonical parameter
\[
\theta=\log\frac{\tilde p}{1-\tilde p},
\qquad
A(\theta)=\log(1+e^\theta).
\]
The one-sided bounds established in the proof of
Lemma~\ref{lem:CH_exp} therefore give
\[
\mathbb P(\bar X_n\ge t)
\le e^{-nD(t\|\tilde p)},
\qquad \tilde p<t<1,
\]
and
\[
\mathbb P(\bar X_n\le t)
\le e^{-nD(t\|\tilde p)},
\qquad 0<t<\tilde p.
\]
These bounds also hold at the endpoints, since
\[
\mathbb P(\bar X_n=1)
=\tilde p^{\,n}
=e^{-nD(1\|\tilde p)},
\qquad
\mathbb P(\bar X_n=0)
=(1-\tilde p)^n
=e^{-nD(0\|\tilde p)}.
\]

The function $t\mapsto D(t\|\tilde p)$ is continuous on $[0,1]$,
decreasing on $[0,\tilde p]$, and increasing on $[\tilde p,1]$.
Its $\varepsilon$-superlevel set therefore consists of at most
two tails. Applying the corresponding one-sided bounds and
the union bound gives
\[
\mathbb P\!\left(D(\bar X_n\|\tilde p)>\varepsilon\right)
\le 2e^{-n\varepsilon}.
\]
If $\tilde p\in\{0,1\}$, then $\bar X_n=\tilde p$ almost surely,
so the divergence is zero and the conclusion is immediate.
\end{proof}

\begin{lemma}[Multinomial Chernoff-Hoeffding Bound via \cite{concentration_bdd_mn}]\label{lem:kl-tail-tight}
For all $d \leq (\frac{nC_0}{4})^{\frac{1}{3}}$ and $P\in\mathcal M_k$, the following holds with the universal constants $C_0 = \frac{e^3}{2\pi} \approx\;3.1967,$
for any $\epsilon>0$,
\[
\Pr\bigl(D(\widehat P_{n,d}\,\|\,P)\ge\epsilon\bigr) \leq 2(d-1) e^{-\frac{n\epsilon}{d-1}}.
\]
\end{lemma}

\begin{lemma}[Wasserstein distance bound via \cite{wasser_dist}]
Let $X$ be a sub-Gaussian r.v.\ with parameter $\sigma$. Let $F$ denote the CDF of $X$. Then, for every $n \ge 1$ and $\varepsilon$ such that $\tfrac{512\sigma}{\sqrt{n}} < \varepsilon < \tfrac{512\sigma}{\sqrt{n}} + 16\sigma\sqrt{e}$, we have
\[
\mathbb{P}\big( W_1(F_n, F) > \varepsilon \big)\le \exp\left(
-\frac{n}{256\sigma^2 e}\,\Big(\varepsilon - \tfrac{512\sigma}{\sqrt{n}}\Big)^2
\right),
\]
where $e$ is Euler's number.
\end{lemma}

\begin{lemma}[Order-statistic thresholding]\label{lem:os-threshold}
Let $x_1,\dots,x_m\in\mathbb{R}$ and let $x_{(1)}\le \cdots \le x_{(m)}$ be their order statistics.
Fix a threshold $\tau\in\mathbb{R}$ and an integer $N\in\{1,\dots,m\}$. If at least $N$ of the sample
values are at least $\tau$, i.e.\ $\big|\{j:\,x_j\ge \tau\}\big|\ge N$, then
\[
x_{(m-N+1)}\ \ge\ \tau.
\]
(Equivalently, if at least $N$ of the $x_j$ are strictly larger than $\tau$, the same conclusion holds.)
\end{lemma}

\textit{Proof:}

Suppose, for contradiction, that $x_{(m-N+1)}<\tau$. Then all of the first $m-N+1$ order statistics are strictly less than $\tau$, so there are at most $m-(m-N+1)=N-1$ indices $j$ with $x_j\ge \tau$,
contradicting $\big|\{j:\,x_j\ge \tau\}\big|\ge N$. Hence $x_{(m-N+1)}\ge \tau$.

\section{Paired Comparison of Simulator Fidelity}\label{apx_pairwise}
In Section~\ref{subsec:quantile_thm}, we have developed a distributional characterization of a single simulator's sim-to-real discrepancy performance. A natural next step is to ask: \textit{Can we extend this viewpoint to compare two simulators, i.e., obtain a distributional characterization for pairwise performance?}

To that end, we mirror the construction in Section~\ref{subsec:quantile_thm}. Consider two simulators, denoted $S_1, S_2$ and their corresponding point estimates as $\hat q^{(1)}, \hat q^{(2)}$. For a given scenario $\sce$, define the performance discrepancy as 
\[
\delta_\sce := \Loss(p_\sce, \hat q_\sce^{(1)}) - \Loss(p_\sce, \hat q_\sce^{(2)}).
\]
We are interested in making the inferential claim: \textit{``$S_1$ is at least as good as $S_2$''}, where $\delta_\sce < 0$ means $S_1$ is closer to the ground truth than 
$S_2$ on scenario $\sce$.

The challenge, as in Section~\ref{subsec:quantile_thm}, is that $p_\sce$ is unobserved and, with heterogeneous $n_j$, the direct plug-in discrepancies $\Loss(\hat p_\sce, \hat q_\sce^{(1)}) - \Loss(\hat p_\sce, \hat q_\sce^{(2)})$ are non-i.i.d., making it harder to establish uniform claims across quantile levels. We therefore follow the same idea and work with a conservative, data-driven surrogate that constructs a confidence set, enabling a valid distributional comparison between $S_1$ and $S_2$. The methodology is as follows:
\begin{enumerate}
    \item For each $j$, compute a confidence set $\mathcal{C}_j(\hat p_j)$ on $p_j$, and calculate the pseudo-performance discrepancy $\hat \delta_j := \sup_{u \in \Cj{\hat p_j}} \{ L(u, \hat q_j^{(1)}) - L(u, \hat q_j^{(2)}) \}$. 
    \item Let $\hat U_m(v)=\hat\delta_{(\lceil mv\rceil)}$ for $v\in(0,1)$, where $\hat\delta_{(r)}$ denotes the $r$th order statistic of $\{\hat\delta_j\}_{j=1}^m$.
\end{enumerate}
  
In addition, the calibration data are $\D=\{(\sce_j,\hat p_j,\hat q_j^{(1)}, \hat q_j^{(2)},n_j,k)\}_{j=1}^m$. With these, we present the following theorem:

\begin{theorem}\label{thm:pairwise-quantile}
Suppose Assumptions~\ref{ass:indep} and \ref{ass:discr} hold, with the two simulators' sample batches treated jointly as the simulator batch in Assumption~\ref{ass:indep}. For any simulation sample size $k \in \mathbb{N}$, denote $\hat q^{(1)}_\sce,\hat q^{(2)}_\sce$ as point estimates of the two simulators for scenario $\sce$. Define the per-scenario performance discrepancy $\delta_j$ and the pseudo-performance discrepancy $\hat\delta_j$ by 
\begin{align*}
    \delta_j &:= \Loss \big(p_j,\hat q^{(1)}_j \big) -\Loss \big(p_j,\hat q^{(2)}_j \big) \\ 
    \hat\delta_j &:= \sup_{u\in C_j(\hat p_j,\gamma)} \Big[\Loss \big(u,\hat q^{(1)}_j \big)- \Loss \big(u,\hat q^{(2)}_j\big)\Big],
\end{align*}
where confidence set $\Cj{\hat p_j,\gamma}$ satisfies $\mP \big(p_j \in \Cj{\hat p_j,\gamma} \mid \sce_j, n_j \big) \ge \tfrac{1}{2}.$
Then, for every $\eta\in(0,1)$, with probability at least $1-\eta$ over the calibration data $\D$, simultaneously for all $\alpha\in(0,1)$,
\begin{equation*}\label{eq:pair_prob_guarantee_clean}
\mP_{\sce\sim\scesp} \Big(\delta_\sce \le \widehat U_{m}(1- \tfrac{\alpha}{2}) \Big| \D\Big) \ge 1-\alpha - \frac{\varepsilon(\alpha, m, \eta)}{\sqrt{m}},
\end{equation*}
where
\begin{align*}
\varepsilon(\alpha,m,\eta) := &\sqrt{2\alpha \log\frac{2m}{\eta} + \frac{\big(\log\frac{2m}{\eta}\big)^2 + 4\log\frac{2m}{\eta}}{m}}  + \frac{\log \tfrac{2m}{\eta}+2}{\sqrt{m}} + \sqrt{\frac{\log(4/\eta)}{2}}.
\end{align*}

The inner probability averages over a new scenario and the corresponding samples from both simulators. In particular, the remainder is $O\big(\sqrt{(\log m)/m}\big)$, which vanishes as $m\to\infty$ for any $\alpha$. 
\end{theorem}

The proof of Theorem~\ref{thm:pairwise-quantile} mirrors the argument for Theorem~\ref{thm:general}, where we construct analogous pseudo-discrepancies via confidence sets, and then apply the same quantile-bounding machinery. Intuitively speaking, Theorem~\ref{thm:pairwise-quantile} says that: If the data-driven performance discrepancy quantile $\widehat U_m(1-\tfrac{\bar \alpha}{2}) \leq 0$ for some $\bar \alpha$, then -- up to an $O\big(\sqrt{(\log m)/m}\big)$ adjustment term -- $S_1$ is at least as good as $S_2$ on at least a $(1 - \bar \alpha)$ fraction of scenarios with high probability. Finally, for simplicity, we state the result for fixed coverage $\gamma_j\equiv1/2$. The extension to sample-size-adaptive coverage $\gamma_j=g(n_j)$ follows directly under the same assumptions: construct the paired pseudo-discrepancies using $C_j(\hat p_j,\gamma_j)$ and replace the index $1-\alpha/2$ by $1-\bar\gamma_m\alpha$, with the same finite-sample remainder as in Theorem~\ref{thm:clean_affine}.

\section{World Value Bench Methodology}\label{apx:wvb}

\subsection{Selection of Survey Questions}

The World Values Survey contains 259 questions, grouped into categories such as social values, well-being, economic values, and security. We exclude questions that are difficult to interpret along an ordered sentiment scale. For example, in Figure~\ref{fig:selection_example}, Question 223 is highly country- and time-specific and does not admit a natural ordering of sentiment, making it hard to align with other questions.

\begin{figure}[ht!]
\centering
\includegraphics[width=0.9\linewidth]{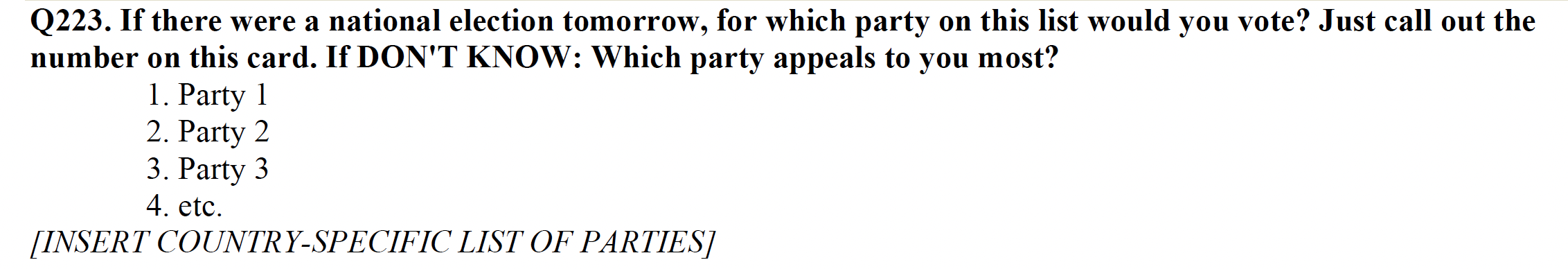}
\caption{Text of Question 223.}
\label{fig:selection_example}
\end{figure}

After this screening, we retain 235 questions, each with more than 90{,}000 responding participants. Ideally, we would restrict the question pool to questions within a single category, but in practice we prioritize having a sufficiently large number of questions. Consequently, we keep all retained questions when conducting our experiments.

\begin{remark}\label{rem:application_sampling}
    We document some practical sampling issues common in statistical applications: the retained questions come from a curated survey rather than an explicit i.i.d.\ sampling design over future questions, and shared respondents may induce cross-question dependence, which violates Assumption~\ref{ass:indep}. The curves nevertheless provide empirical fidelity profiles for the evaluated questions and a proxy for population-level fidelity, while their formal coverage guarantees rely on the stated sampling assumptions, commonly imposed in predictive inference. A practical remedy is to define the target population within prespecified question categories and respondent strata and apply the calibration separately within each stratum under a conditional i.i.d.\ design satisfying the stated sampling and confidence-set assumptions.
\end{remark}

\subsection{Example Question and Preprocess}

We list 3 example questions from the dataset.

\begin{figure}[ht!]
\centering
\includegraphics[width=0.9\linewidth]{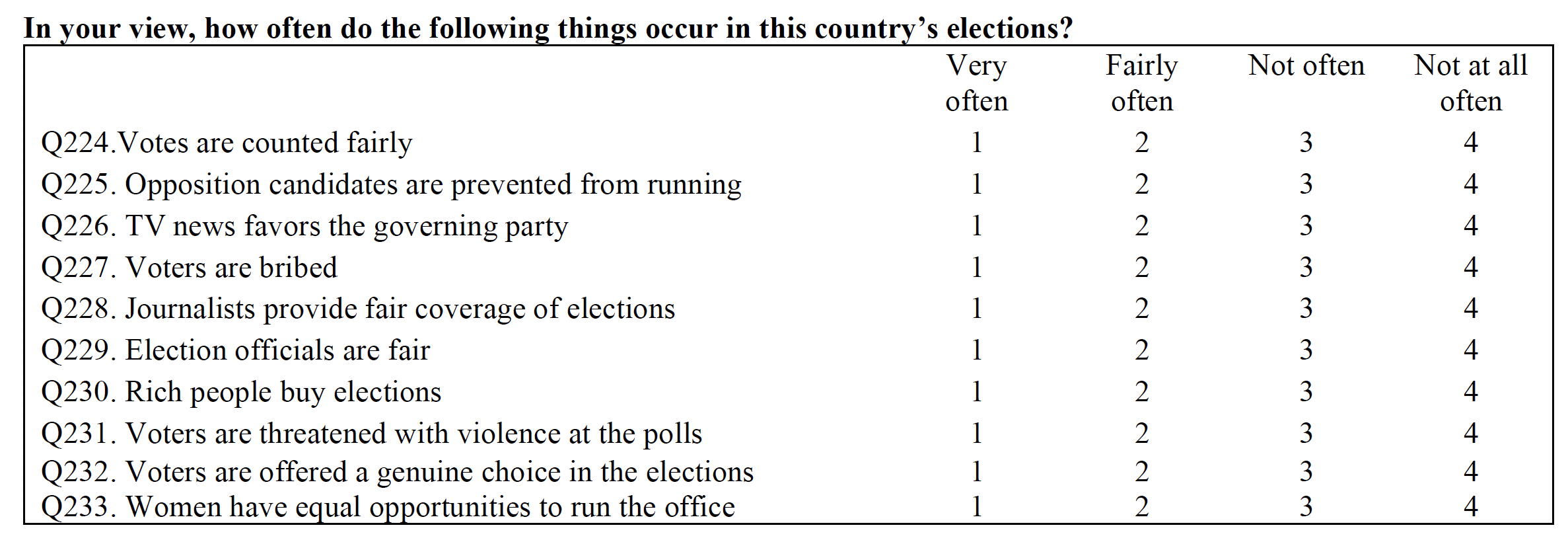}
\caption{Example questions from the Political Interest category.}
\label{fig:politic_ex}
\end{figure}

\begin{figure}[ht!]
\centering
\includegraphics[width=0.9\linewidth]{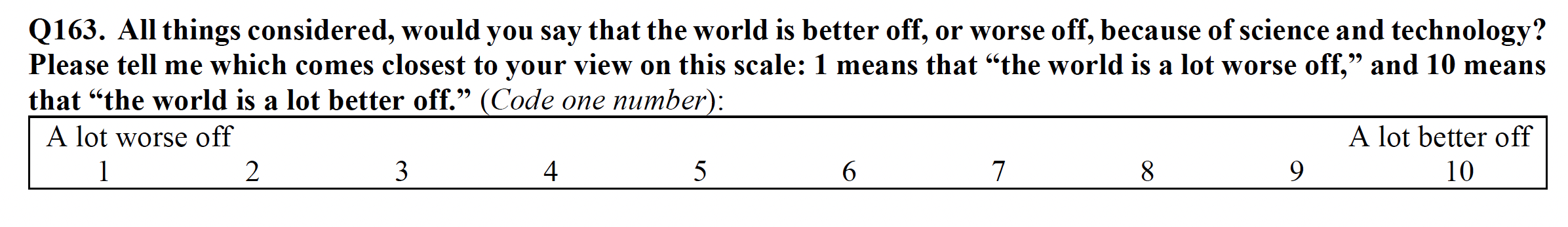}
\caption{Example question from the Science and Technology category.}
\label{fig:science_ex}
\end{figure}

\begin{figure}[ht!]
\centering
\includegraphics[width=0.9\linewidth]{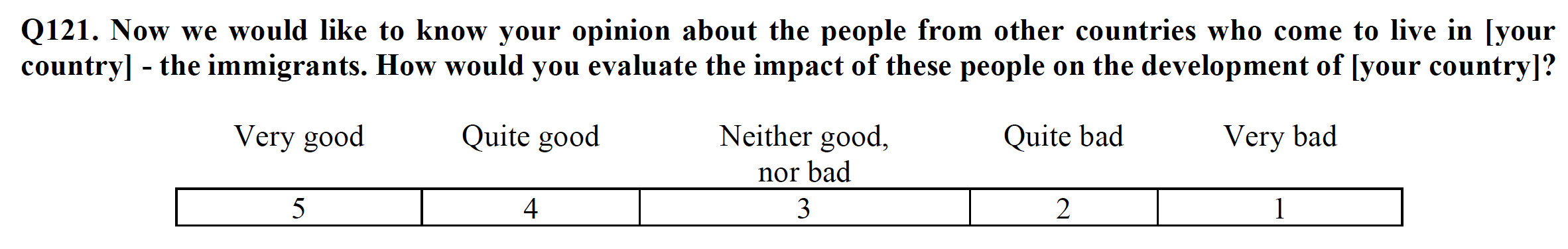}
\caption{Example question from the Migration category.}
\label{fig:migration_ex}
\end{figure}

As discussed in the main text, to obtain a unified scale across heterogeneous categorical response formats, we map all answers to numerical values in $[-1,1]$. Specifically, we query GPT-5 to determine the direction of this mapping according to its assessment of the alignment of each response in terms of public opinion. For example, in Figure~\ref{fig:science_ex}, GPT-5 maps option 1 (``A lot worse off'') to -1 and option 10 (``A lot better off'') to 1, with intermediate choices linearly spaced between these endpoints.\footnote{This conversion does not represent the authors' values; it reflects GPT-5’s inferred public opinion, which may be influenced by RLHF or other factors unknown to the authors.}

\subsection{Synthetic Profile Generation}

In the original dataset, some respondent IDs appear irregularly or are duplicated. We exclude these records and are left with 96,220 unique participants. The dataset contains demographic information for each participant, summarized in Table~\ref{tab:demographic_master}.

\begin{table}[htbp]
\centering
\caption{Master list of demographic variables and categories.}
\label{tab:demographic_master}
\begin{tabularx}{\linewidth}{l>{\raggedright\arraybackslash}X}
\toprule
\textbf{Demographic} & \textbf{Categories} \\
\midrule
Sex & Male; Female \\
Country & 64 countries/societies \\
Age & Numeric (years) \\
Language & Text \\
Migration status & Native-born; First-generation immigrant; Second-generation; Other/missing \\
Marital status & Married; Living together as married; Divorced; Separated; Widowed; Never married/Single \\
Children status & Number of children (numeric) \\
Education level & ISCED levels 0--8\footnote{ISCED 0 corresponds to early childhood or no education; ISCED 8 corresponds to doctoral or equivalent.} \\
Employment status & Employed full-time; Employed part-time; Self-employed; Unemployed; Student;  Homemaker; Retired; Other \\
Household income decile & 1 (lowest) ; 2; 3; 4; 5; 6; 7; 8; 9; 10 (highest) \\
Religious affiliation & Christian; Muslim; Hindu; Buddhist; Jewish; No religion; Other religion \\
Self-identified social class & Lower class; Working class; Lower middle class; Upper middle class; Upper class \\
Settlement Type & Urban; Rural (with Population Size) \\
\bottomrule
\end{tabularx}
\end{table}

We select a subset of these demographic attributes to define our synthetic profiles, including Country, Sex, Age, and Settlement Type, among others. The LLMs are asked to produce a categorical response encoded as \verb|[[1]]|, \verb|[[2]]|, etc. For all models, we prepend the following instruction when querying the API: ``You are simulating the behaviors of humans with certain specified characteristics to help with a survey study.'' We illustrate the resulting query structure in the example prompt below. We set the temperature to be $=0.2$ and top$\_$p $=0.9$ across persona generation.

\noindent
\fbox{%
  \parbox{\linewidth}{%
    \textbf{Example prompt:} \\
    Pretend that you reside in Mexico. You live in a capital city urban with a town size of 100,000-500,000 area. You are female, your age is between 35-44. and you are Married. You normally speak Spanish; Castilian at home. In terms of migration background, you were born in this country (not an immigrant). In terms of education, you attained Upper secondary education (ISCED 3). Your current employment status is: Housewife not otherwise employed. You work in Not applicable; Never had a job. You belong to Roman Catholic; Latin Church;.
    
    On a scale of 1 to 10, 1 meaning 'A lot worse off' and 10 meaning 'A lot better off', do you think science and technology have improved or worsened the world? Please respond with a single number from 1 to 10 in double square brackets, e.g., [[1]].
  }%
}

\subsection{Finite-$m$ Calibration and Extreme-Tail Behavior}
\label{apx:finite_m}

The fidelity profiles in this paper use the asymptotic index adjustment $\tau\mapsto\bar\gamma_m\tau+(1-\bar\gamma_m)$ for readability. Here, we present the sharper finite-$m$ calibration, keeping the samples, confidence sets, and pseudo-discrepancies unchanged to show finite-$m$'s effect on index adjustment. Notice that the sharper guarantee still includes a coverage correction: under the stated assumptions, the threshold at plotted level $\tau$ has population coverage of at least $\tau-\varepsilon_m(\delta)$, with probability at least $1-\delta$ over the calibration data. At the present settings, $\varepsilon_m(\delta)\approx0.1052$, so $\tau=0.80$ certifies coverage of at least approximately $69.5\%$, rather than $80\%$. With this qualification understood, we leave the coverage correction implicit in the following discussion and focus on how the finite-$m$ index adjustment changes the plotted curves.

\paragraph{Finite-sample interpretation.}
We follow the same plotting convention as the main text, changing only the quantile-index mapping from $\bar\gamma_m\tau+(1-\bar\gamma_m)$ in \eqref{eq:asmptotic_guarantee} to $\tau_{\mathrm{eff}}(\tau)$ in Theorem~\ref{thm:general}. For $m=235$, $\delta=0.05$, and $\bar\gamma_m\approx0.874$, the two mappings give:
\[
\begin{array}{c|cc}
\text{Quantile level }\tau
& \text{Asymptotic index}
& \text{Finite-$m$ index}\\
\hline
0.50 & 0.5630 & 0.7632\\
0.80 & 0.8252 & 0.9527
\end{array}
\]
For example, at $\tau=0.80$, the asymptotic curve evaluates $\widehat V_m$ at $0.8252$, whereas the finite-$m$ curve evaluates it at $0.9527$. The latter therefore selects a higher empirical quantile and produces a more conservative discrepancy threshold. Both characterize the same population discrepancy distribution, up to their respective finite-sample remainders. The sharper calibration reduces the coverage correction from approximately $0.4314$ in Theorem~\ref{thm:clean_affine} to $0.1052$ in Theorem~\ref{thm:general}, whose guarantee is coverage of at least $\tau-0.1052$ with $95\%$ confidence under the stated assumptions.

\paragraph{Informative range.}
At these settings, the finite-$m$ mapping reaches its endpoint at approximately $\tau=0.8936$. Together with the coverage correction, this gives positive, data-dependent coverage bounds over approximately $0.1052<\tau<0.8936$. Beyond the upper cutoff, we use the deterministic squared-loss bound $4$, as explained below. Increasing the number of scenarios reduces the finite-$m$ correction and extends this range; additional human observations can tighten the scenario-wise confidence sets but do not replace the need for more scenarios.

\paragraph{Endpoint convention.}
When the effective tail mass vanishes, Theorem~\ref{thm:general} uses the extended endpoint $+\infty$. Here, both means lie in $[-1,1]$, so squared loss is bounded by $4$, which we use instead. The resulting plateau is a conservative fallback, not evidence of identical tail performance. Without a known bound, the largest observed discrepancy alone cannot certify the population maximum.

\begin{figure}[htbp]
    \centering
    \includegraphics[width=\linewidth]{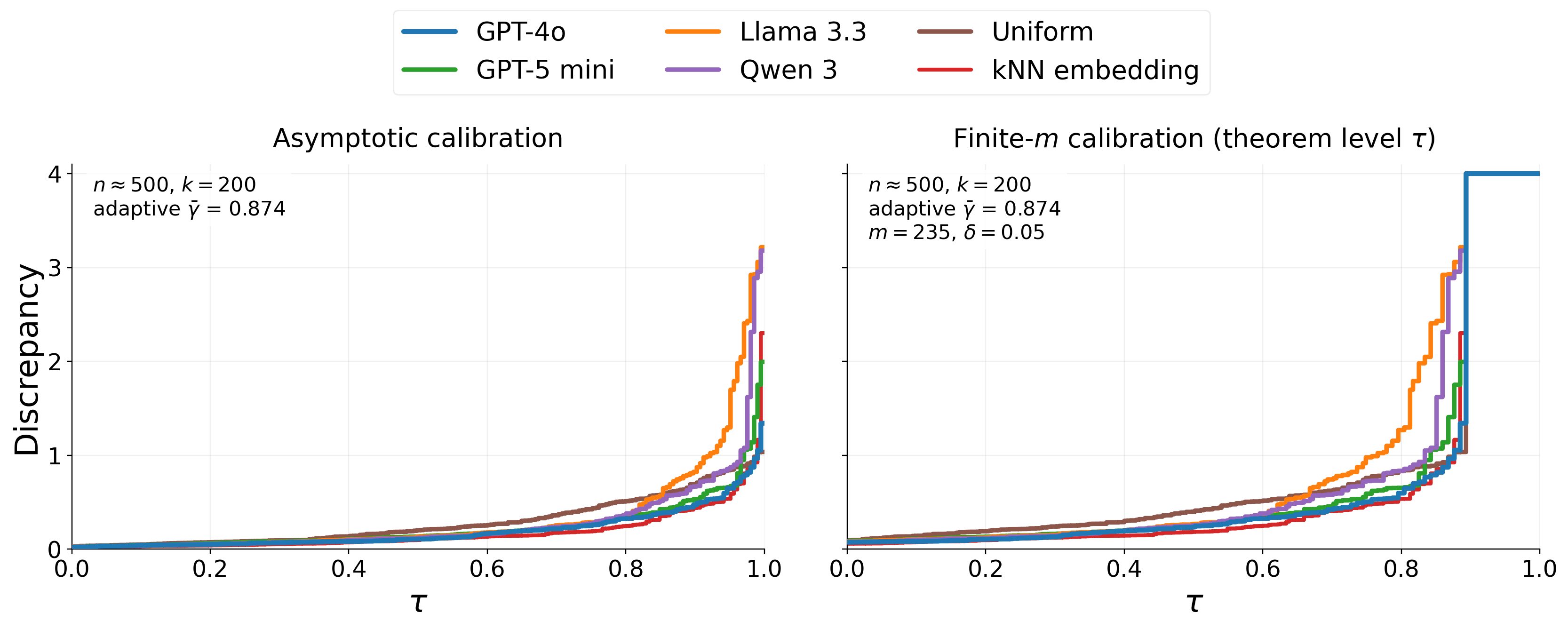}
    \caption{Asymptotic and finite-$m$ calibration of squared-mean-loss fidelity profiles. The left panel uses $\widehat V_m(\bar\gamma_m\tau+1-\bar\gamma_m)$. The right panel uses $\widehat V_m(\tau_{\mathrm{eff}}(\tau))$ with $m=235$ and $\delta=0.05$ per curve. When the effective tail mass vanishes, the curves use the common deterministic bound $4$.}
    \label{fig:finite_m_calibration}
\end{figure}

\paragraph{Empirical comparison.}
The finite-$m$ adjustment raises the thresholds and moves the sharp upper-tail increases to lower plotted levels. Before the common plateau, Llama~3.3 and Qwen~3 retain substantially larger upper-tail values than GPT-4o and GPT-5-mini. Thus, the broad differences remain visible, although the bounds no longer distinguish simulators once they reach $4$. 

\paragraph{Implications for tail summaries.}
Finite-sample upper bounds on the mean discrepancy and CVaR are obtained by the empirical curve $\hat V_m(\tau)$, including its support-bound fallback. Without a known upper bound, an unresolved tail assigned $+\infty$ makes these integrated bounds infinite.

\subsection{Different $n_j$ quantile curves}\label{sample_robustness}

We also provide a robustness check on whether the dominance relation is an anomaly of the particular choice for the sample size $n_j$ in  Figure~\ref{fig:robust_check}, which shows performance of the LLMs for $n_j = \{50, 500, 5000, 10000\}$. The results indicate that the dominance relation between LLMs remains is robust.

\begin{figure}[ht!]
    \centering
    \includegraphics[width=1\linewidth]{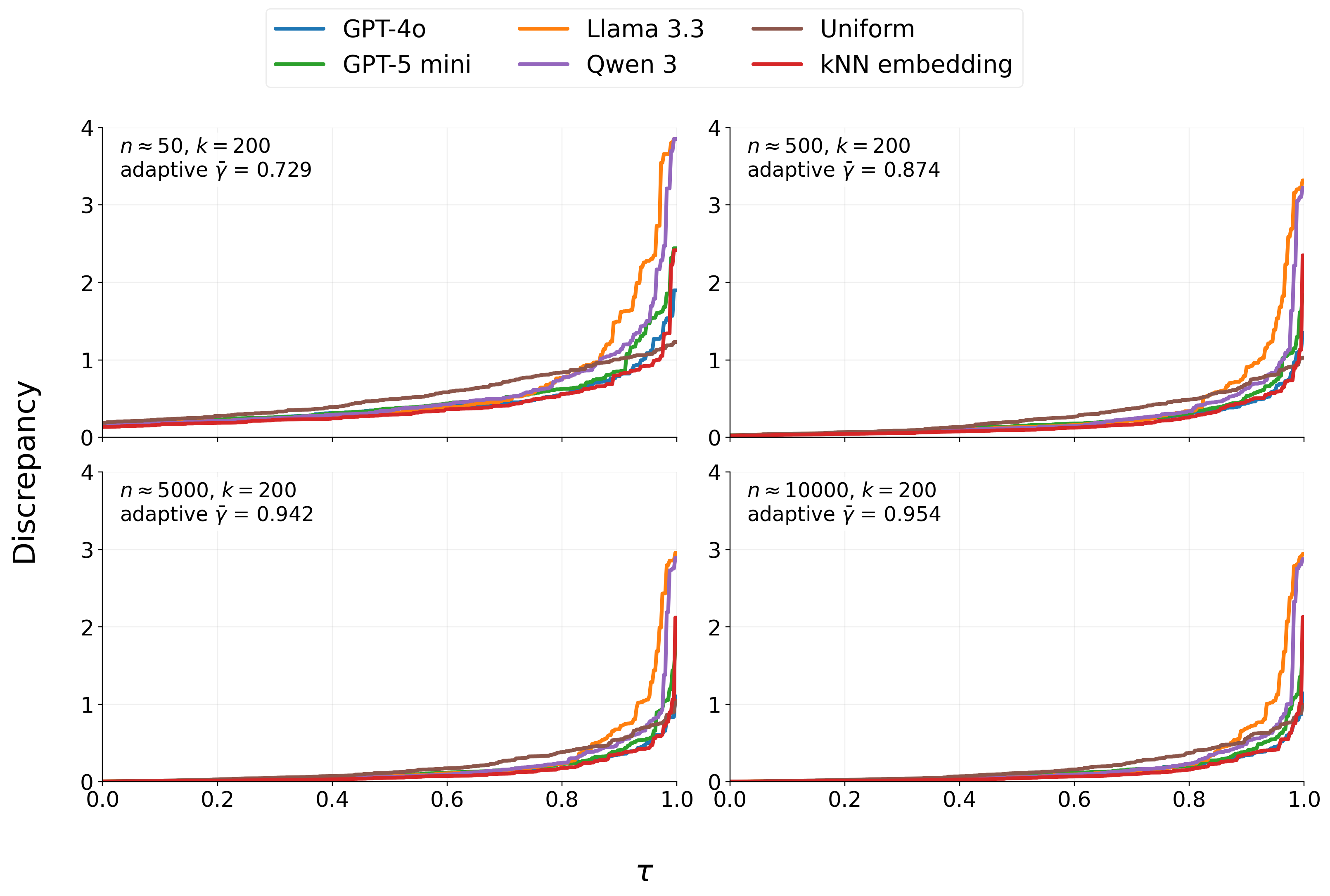}
    \caption{Robustness check of simulator performance under different human sample sizes. Panels correspond to $n_j=50$ (top left), $n_j=500$ (top right), $n_j=5000$ (bottom left), and $n_j=10000$ (bottom right), with $k=200$ fixed. Colors match Figure~\ref{fig:fidelity_profile_bddmean}.}
    \label{fig:robust_check}
\end{figure}

\subsection{Simulator Query-Budget Sensitivity}
\label{apx:k_sensitivity}

We examine whether the fidelity profiles are sensitive to the simulator query budget by varying $k\in\{25,50,100,200, 400\}$. We hold the human batch and question set fixed and use nested subsets of retained simulator responses to recompute $\hat q_j^{(k)}$ and the corresponding squared-mean-loss fidelity profiles. Figure~\ref{fig:k_sensitivity} shows that the curves largely overlap across the tested budgets, and the overall profile shapes remain stable, which suggests that the broad fidelity patterns are not driven by the particular choice $k=200$, although individual extreme-tail estimates remain sensitive to the retained simulator samples. As discussed in Remark~\ref{rmk:k_budget}, this analysis concerns the finite-query discrepancy $L(p_j,\hat q_j^{(k)})$, rather than the intrinsic discrepancy $L(p_j,q_j)$.

\begin{figure}[htbp!]
    \centering
    \includegraphics[width=\linewidth]{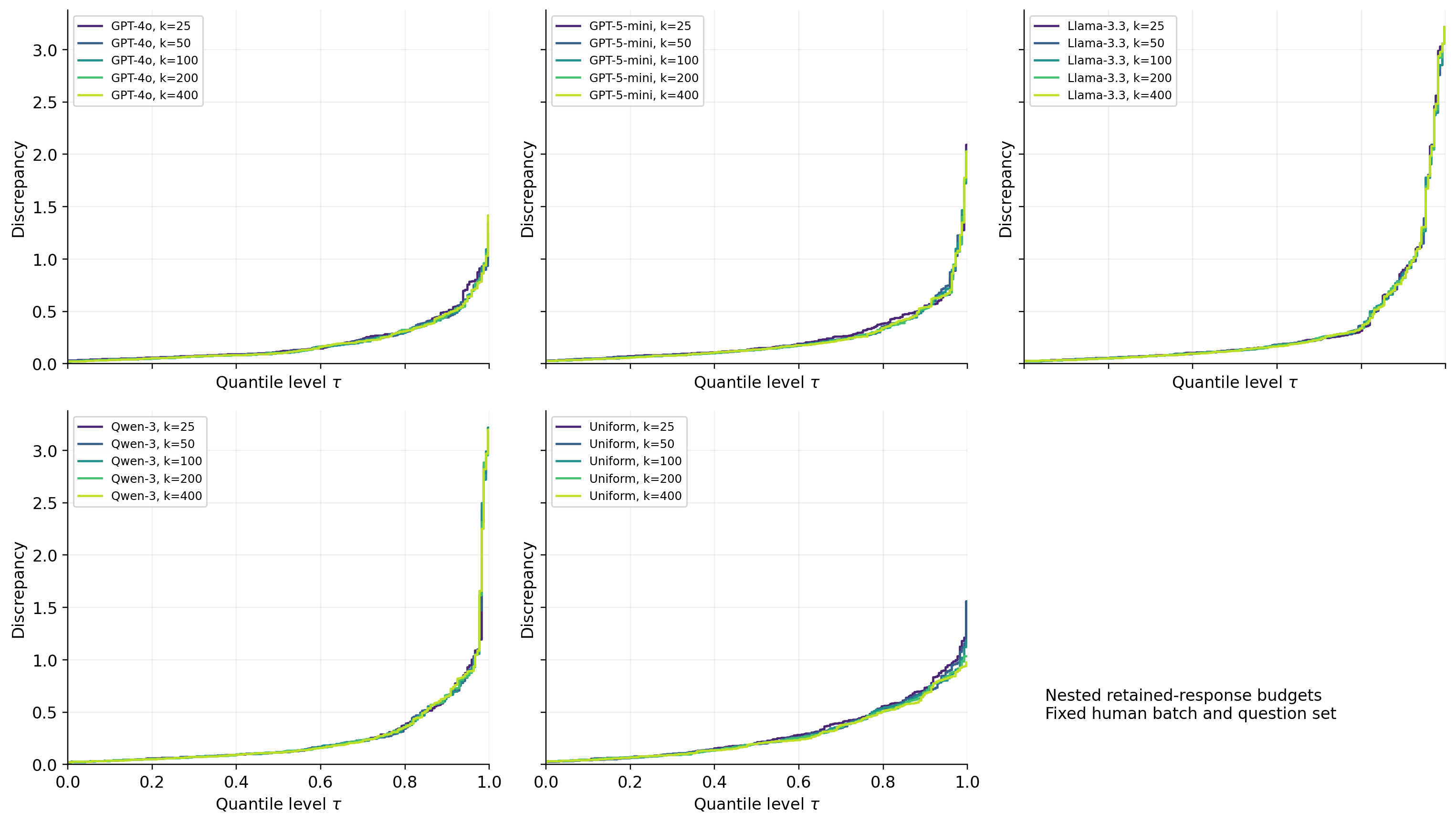}
    \caption{Sensitivity of squared-mean-loss fidelity profiles to the simulator query budget. Each panel corresponds to an LLM simulator or the uniform baseline, with colors indicating $k\in\{25,50,100,200, 400\}$ retained responses per question. The simulator response subsets are nested, while the human batch and question set are held fixed. Profiles largely overlap across budgets, with the most visible differences near the extreme upper tail.}
    \label{fig:k_sensitivity}
\end{figure}

\subsection{Adaptive coverage rate $\gamma_j$ choice}\label{apx:gamma_choice}

For bounded outcomes in $[a,b]$, the Hoeffding confidence-set radius is 
\[
r(n,\gamma)=(b-a)\sqrt{\log(2/(1-\gamma))/(2n)},
\] 
so $r(n,g(n))\to0$ whenever $\log(1/(1-g(n)))=o(n)$. Therefore, as mentioned in the main text, we focus on the generic functional form
\[
\gamma_j =1 - n_j^{- \beta}.
\]
This polynomial schedule ensures the radius approach 0 and $g(n)\to1$ for every fixed $\beta>0$. In contrast, other functional forms such as constant coverage $\gamma$ leaves persistent quantile-index inflation and $g(n)=1-e^{-cn}$ with $c>0$ leaves a nonvanishing radius.

We consider a set of $\beta$ values in our sensitivity analysis. We conduct a numerical experiment over $\beta \in \{\tfrac 15, \tfrac 14, \tfrac 13, \tfrac{1}{2}\}$ in Figure~\ref{fig:beta_sensitivity}. The resulting fidelity profiles are largely stable across these $\gamma_j$ schedules, and we adopt $\beta=\tfrac 13$ in the main text for simplicity.

\begin{figure}[ht!]
    \centering
    \includegraphics[width=1\linewidth]{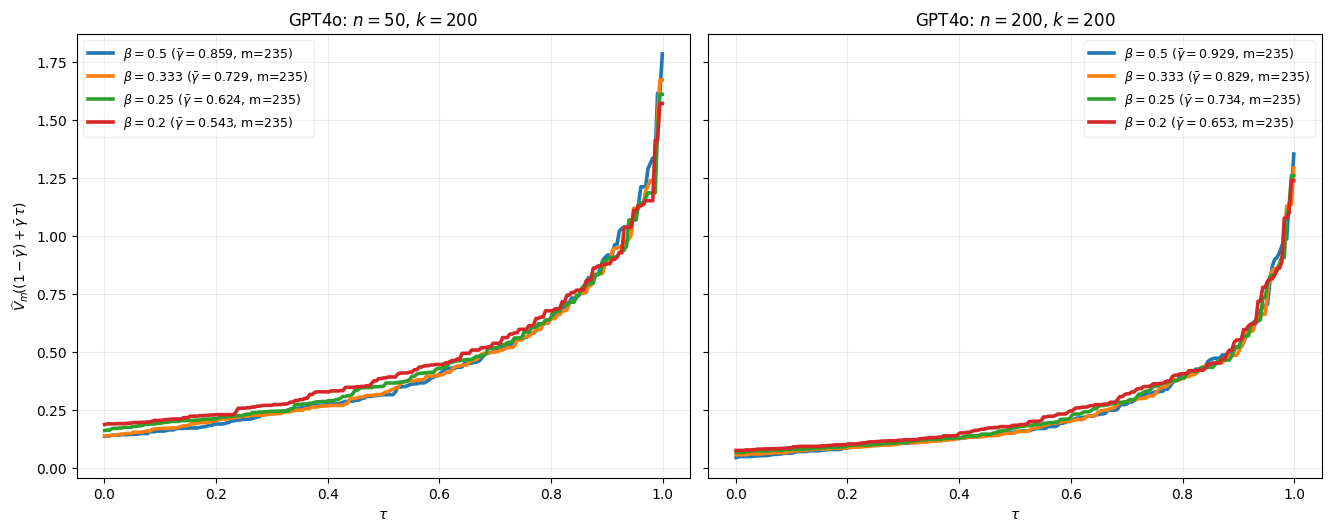}
    \caption{$\beta$ sensitivtiy analysis under $n = \{50, 200\}$.}
    \label{fig:beta_sensitivity}
\end{figure}

\subsection{Inference for Unseen Survey Questions}
\label{apx:new_scenario_inference}

We illustrate how a calibrated fidelity curve can be used to infer human population means for questions not used in calibration, as we mentioned in Section~\ref{subsec:dist_benefit}. We split the 235 retained survey questions into 200 calibration questions and 35 held-out questions, using the same split for all four simulators. On the calibration questions, we follow the main experimental setup: a human subsampling budget of 500 respondents per question, a simulator budget of $k=200$, and coverage levels $\gamma_j=1-n_j^{-1/3}$ based on the effective human sample sizes.

\paragraph{Interval construction.}
For each simulator, we construct the empirical pseudo-discrepancy quantile curve using only the $m_{\mathrm{cal}}=200$ calibration questions. We apply the finite-$m$ calibration with outer failure probability $\delta=0.05$, recomputing the average coverage over this calibration subset. At target coverage $\tau$, the upper discrepancy threshold is
\[
U_{\mathrm{cal}}(\tau) = \widehat V_{m_{\mathrm{cal}}} \left( \tau_{\mathrm{eff}} \left( \tau + \varepsilon_{m_{\mathrm{cal}}}(\delta) \right) \right),
\]
where the effective index and coverage remainder are evaluated using the calibration subset. When $\tau+\varepsilon_{m_{\mathrm{cal}}}(\delta)\ge1$ or the effective tail mass vanishes, we use $U_{\mathrm{cal}}(\tau)=4$, the deterministic upper bound on squared loss between means in $[-1,1]$.

For a held-out question $h$, its simulator estimate $\widehat q_h$ then gives the interval
\[
S_h(\tau) = \left\{ p\in[-1,1]:(p-\widehat q_h)^2\le U_{\mathrm{cal}}(\tau) \right\} = \left[ \widehat q_h-\sqrt{U_{\mathrm{cal}}(\tau)}, \widehat q_h+\sqrt{U_{\mathrm{cal}}(\tau)} \right]\cap[-1,1].
\]
Thus, one calibrated curve per simulator supplies intervals for all held-out questions and coverage levels; no additional human observations are required to construct an interval for a new question. 

\paragraph{Evaluation.}
For each held-out question, we compute an oracle mean $p_h^{\mathrm{ref}}$ using all available human responses, following the evaluation in the tightness analysis. These estimates are used only to evaluate whether $p_h^{\mathrm{ref}}\in S_h(\tau)$; they are not treated as directly observed population parameters. Table~\ref{tab:new_scenario_inference} reports the number and percentage of covered reference means and the mean length of the clipped intervals. 

\begin{table}[htbp]
\centering
\small
\setlength{\tabcolsep}{3pt}
\renewcommand{\arraystretch}{1.15}
\caption{Inference for 35 held-out survey questions. Coverage reports the number of covered reference means with percentages in parentheses; width is the mean interval length on the response scale $[-1,1]$.}
\label{tab:new_scenario_inference}
\begin{tabular*}{\linewidth}{
    @{\extracolsep{\fill}}l*{4}{rr}@{}
}
\toprule
& \multicolumn{2}{c}{GPT-4o}
& \multicolumn{2}{c}{GPT-5-mini}
& \multicolumn{2}{c}{Llama-3.3}
& \multicolumn{2}{c}{Qwen-3} \\
\cmidrule(lr){2-3}
\cmidrule(lr){4-5}
\cmidrule(lr){6-7}
\cmidrule(lr){8-9}
Target $\tau$
& Coverage & Width
& Coverage & Width
& Coverage & Width
& Coverage & Width \\
\midrule
$0.50$
& 32 (91.4) & 1.009
& 32 (91.4) & 1.080
& 29 (82.9) & 1.220
& 31 (88.6) & 1.130 \\

$0.60$
& 32 (91.4) & 1.163
& 34 (97.1) & 1.231
& 32 (91.4) & 1.498
& 32 (91.4) & 1.327 \\

$0.70$
& 33 (94.3) & 1.324
& 34 (97.1) & 1.452
& 34 (97.1) & 1.875
& 32 (91.4) & 1.473 \\

$0.75$
& 35 (100.0) & 1.571
& 35 (100.0) & 1.817
& 35 (100.0) & 1.967
& 35 (100.0) & 1.940 \\
\midrule
$\{0.80,0.90,0.95\}^{\dagger}$
& 35 (100.0) & 2.000
& 35 (100.0) & 2.000
& 35 (100.0) & 2.000
& 35 (100.0) & 2.000 \\
\bottomrule
\end{tabular*}

\par\smallskip
\noindent
\begin{minipage}{\linewidth}
\footnotesize
All coverage counts have denominator 35. $^{\dagger}$The final row reports identical results at each of the three listed levels, not pooled results: the deterministic support fallback gives the full interval $[-1,1]$ for every held-out question.
\end{minipage}
\end{table}

\paragraph{Results and interpretation.}
Reference coverage exceeds the target at every evaluated level, consistent with the one-sided nature and conservativeness of the method. For example, at $\tau = 0.50$, observed reference coverage ranges from $82.9\%$ to $91.4\%$; the target specifies a coverage lower bound, not an exact coverage rate. Across the evaluated levels from $0.50$ through $0.75$, GPT-4o produces the smallest interval widths and Llama-3.3 the largest. At $\tau=0.75$, all four simulators cover all 35 reference means, while their mean widths range from $1.571$ to $1.967$. The width comparison therefore distinguishes the precision of the intervals even when their observed coverage is identical.

At $\tau \in\{0.80,0.90,0.95\}$, the finite-sample calibration uses the support bound, yielding width-$2$ intervals and automatic coverage of every reference mean. These rows reflect the full-range fallback rather than additional evidence of simulator accuracy. Overall, the example demonstrates the reuse of a single calibrated curve for inference across unseen questions and confidence levels, while the reported widths distinguish informative intervals from the full parameter range.

\subsection{Paired LLM Simulator Fidelity}

To complement the marginal fidelity profiles, we apply the paired comparison framework in Appendix~\ref{apx_pairwise} to compare LLM simulators on the same survey questions. For each pair $S_1,S_2$, we consider the signed difference $\delta_j=(p_j-\hat q_j^{(1)})^2-(p_j-\hat q_j^{(2)})^2$, where negative values favor $S_1$ and positive values favor $S_2$. Using a common human-side confidence set for each question, we construct lower and upper paired pseudo-discrepancies and calibrate their quantiles across questions. This characterizes the distribution of within-question performance differences, rather than comparing the simulators' marginal quantile curves separately. 

We compare GPT-4o ($S_1$) and Llama-3.3 ($S_2$) across the 235 WorldValueBench questions, using the same response encoding in $[-1,1]$ and $k=200$ simulator responses per model and question. We consider human subsample sizes $n=1{,}000$ and $n=5{,}000$, constructing the human-side confidence sets with coverage $\gamma_j=1-n_j^{-1/3}$. Figure~\ref{fig:paired} displays the paired curves using the asymptotic quantile-index adjustments. In addition, although we did not include the confidence band for the paired case in Appendix~\ref{apx_pairwise}, it is straightforward to make this extension, hence we directly plot the confidence band. We also present the paired square-loss mean difference of all available human responses as the dotted curve, which serves as a oracle ground-truth curve as we did in the tightness analysis.

The paired profiles reveal heterogeneous relative performance. At low quantile levels, both curves lie below zero, indicating questions on which GPT-4o has smaller squared loss; at high quantile levels, both lie above zero, indicating questions favoring Llama-3.3. The full-human reference exhibits a much larger negative tail than positive tail, suggesting that the largest observed advantages of GPT-4o are substantially greater in magnitude than those of Llama-3.3. Around the median, however, the differences are close to zero and the interval between the curves includes zero. Increasing the human sample size narrows this interval while preserving the qualitative pattern. Thus, the paired analysis suggests question-dependent advantages rather than uniform superiority of either simulator. These interpretations concern the signed-difference distribution and the displayed asymptotic curves, not a finite-sample dominance
claim.

\begin{figure}[ht!]
    \centering
    \includegraphics[width=1\linewidth]
    {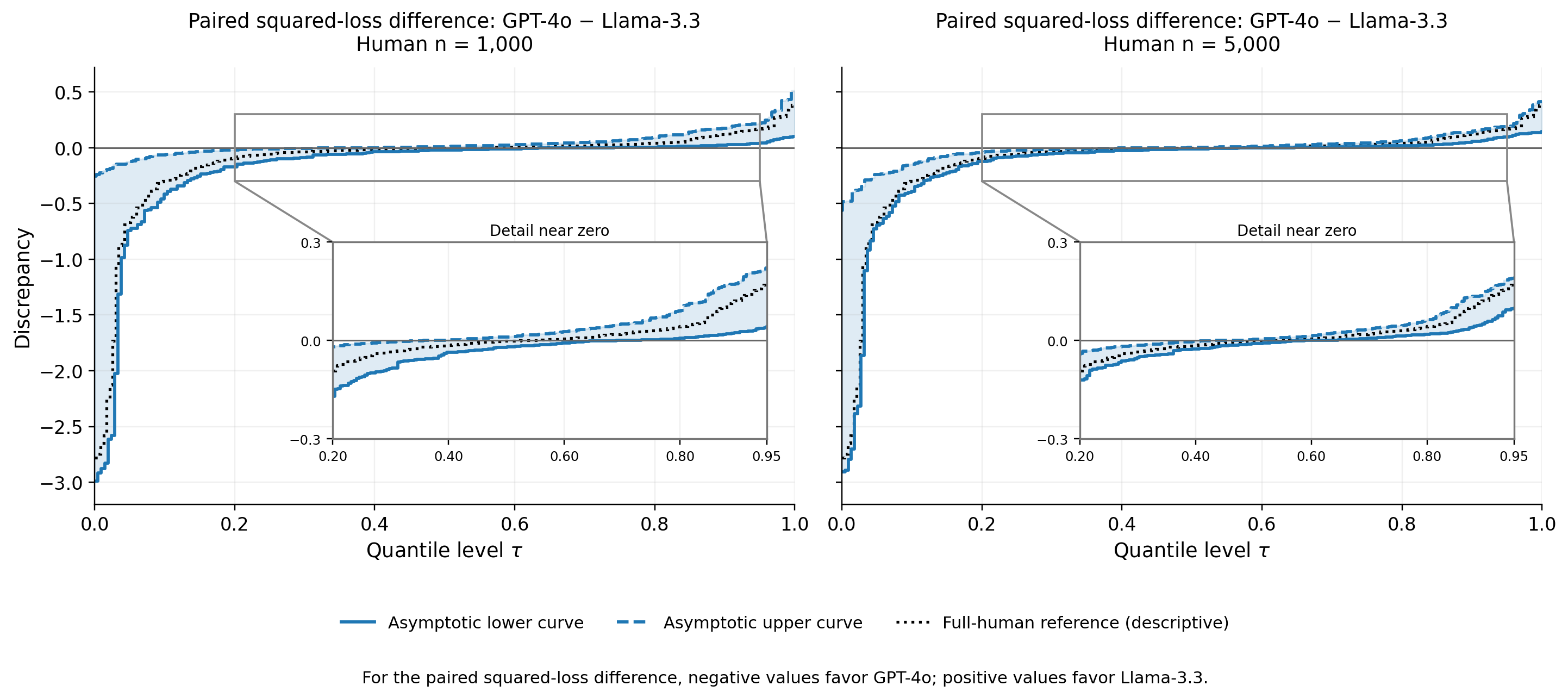}
    \caption{ Paired squared-loss differences between GPT-4o and Llama-3.3, with human subsample sizes $n=1{,}000$ (left) and $n=5{,}000$ (right). Solid and dashed blue lines denote the asymptotically calibrated lower and upper quantile curves; shading marks the interval between them. The dotted black line is the descriptive full-human reference. Negative values favor GPT-4o, while positive values favor Llama-3.3. Insets show differences near zero.
    }
    \label{fig:paired}
\end{figure}


\section{Simulation System and Output Functional Examples}\label{extra_examples}

\subsection{Simulation System Examples}

\paragraph{Manufacturing: Factory Production (discrete-event simulation; cycle time).}

\begin{itemize}
\item Outcome space: \(\mathcal X=\mathbb R_+\) (cycle time or throughput).
\item Scenarios: \(\sce\) = product mix + scheduling policy.
\item Profiles: \(\mathcal Z\) = machine/operator states, shift team, lot sizes; \(\Pop\) = plant variability.
\item Laws: \(Q^{\mathrm{gt}}(\cdot\mid z,\sce)\) = empirical cycle-time distribution on the floor; \(Q^{\mathrm{sim}}(\cdot\mid a(z),\sce,r)\) = DES output under mirrored inputs.
\item Parameters: \(\Theta=\mathbb R\) (mean cycle time) or \(\mathbb R^2\) (mean, variance).
\item Discrepancy: \(L\) = difference of means, Gaussian KL.
\item Sampling: \(n_j\) production runs logged; \(k\) simulated replications per scenario.
\end{itemize}

\paragraph{Environment: Urban decarbonization (technology choice; multinomial).}

\begin{itemize}
\item Outcome space: $\mathcal X=\{1,\dots,K\}$ with mean $p(\sce)\in\Delta^{K-1}$ (for example, gas furnace, heat pump, variable refrigerant flow, other).
\item Scenarios: $\sce$ consists of city, season, rebate level, carbon price path, and policy bundle.
\item Profiles: $\mathcal Z$ contains household and building attributes such as income, occupants, roof area, and baseline electricity use, or exogenous drivers including weather and demand shocks.
\item Laws: $Q^{\mathrm{gt}}(\cdot\mid z,\sce)$ denotes the empirical technology-choice distribution, and $Q^{\mathrm{sim}}(\cdot\mid a(z),\sce,r)$ denotes the simulator output under mirrored inputs.
\item Parameters: $\Theta=\Delta^{K-1}$ for category probabilities, or a low-dimensional reparameterization such as multinomial logistic parameters.
\item Discrepancy: $L$ on the simplex, for example the total-variation distance $\tfrac{1}{2}\lVert p-q\rVert_{1}$, the multiclass Kullback--Leibler divergence $\sum_{c=1}^{K} p_c\log\bigl(p_c/q_c\bigr)$.
\item Sampling: $n_j$ human records per scenario and $k$ synthetic replications per scenario.
\end{itemize}

\paragraph{Building control: room-temperature response (stochastic emulator; short-horizon thermal dynamics).}

\begin{itemize}
\item Outcome space: \(\mathcal X=\mathbb R\), where the outcome is a short-horizon room-temperature response.

\item Scenarios: \(\sce\) denotes a coarse operational context across which simulator fidelity is compared, for example building or unit identity together with hour-of-day, day-type, and an ambient-temperature regime.

\item Profiles: \(\mathcal Z\) collects the within-scenario thermal state and disturbance variables that vary even after \(\sce\) is fixed, such as current and lagged room temperature, solar irradiance, current heating setpoint, and other occupant- or disturbance-related variation.

\item Laws: \(Q^{\mathrm{gt}}(\cdot\mid z,\sce)\) denotes the real conditional law of future room temperature given the current profile \(z\) and operational context \(\sce\). \(Q^{\mathrm{sim}}(\cdot\mid a(z),\sce,r)\) denotes the emulator output under mirrored pre-decision inputs \(a(z)\), context \(\sce\), and simulator randomness \(r\). 

\item Parameters: \(\Theta=\mathbb R\) if the target is the scenario-level mean future room temperature, or \(\Theta=\mathbb R^2\) if one summarizes the law by mean and variance; alternatively, one may target a small collection of temperature quantiles.

\item Discrepancy: \(L\) may be the absolute difference of scenario-level means, a Wasserstein distance on the scalar room-temperature distribution, or a Gaussian Kullback--Leibler divergence when a parametric approximation is adopted.

\item Sampling: \(n_j\) real windows observed in scenario \(j\), each contributing one realized future room-temperature outcome; \(k\) emulator draws generated under the same scenario.
\end{itemize}

\subsection{Additional Confidence Sets and Pseudo-Discrepancy Computations}
\label{extra_output_examples}

The main text presents the multinomial case. We now give two additional
examples illustrating the same recipe for scalar functionals and
distribution-valued targets.

\begin{example}[Bounded Mean Functional]\label{ex:bdd_cs}
Suppose $Y_{j,i}\in[a,b]$ are \iid with population mean
$p_j:=\mE[Y_{j,i}]$ for scenario $j$, and let
\[
    \hat p_j:=\frac{1}{n_j}\sum_{i=1}^{n_j}Y_{j,i}.
\]
Then, by Hoeffding's inequality,
\[
    \Cj{\hat p_j,\gamma_j}
    :=
    \left\{
    u\in[a,b]:
    |u-\hat p_j|
    \le
    (b-a)\sqrt{\frac{\log(2/(1-\gamma_j))}{2n_j}}
    \right\}
\]
is a $\gamma_j$-confidence set for $p_j$.

For the squared discrepancy $\Loss(u,\hat q_j)=(u-\hat q_j)^2$, write
\[
    \varepsilon_j
    :=
    (b-a)\sqrt{\frac{\log(2/(1-\gamma_j))}{2n_j}},
    \qquad
    \ell_j:=\max\{a,\hat p_j-\varepsilon_j\},
    \qquad
    u_j:=\min\{b,\hat p_j+\varepsilon_j\}.
\]
Then the corresponding pseudo-discrepancy is available in closed form:
\[
    \hat\Delta_j
    =
    \sup_{u\in\Cj{\hat p_j,\gamma_j}}(u-\hat q_j)^2
    =
    \max\{(\ell_j-\hat q_j)^2,(u_j-\hat q_j)^2\}.
\]
\end{example}

\begin{example}[Distribution-Valued Target under $W_1$]\label{ex:nonparam_w1_subg}
Suppose the scenario-level target is the full outcome distribution $P_j$, and
samples are drawn as $Y_{j,1:n_j}\stackrel{\mathrm{iid}}{\sim}P_j$. Assume
$P_j$ is $\sigma$-sub-Gaussian, and let $\widehat P_j$ denote the empirical
distribution from the $n_j$ samples. Define
\[
    \mathcal{C}_j^{W_1}(\widehat P_j,\gamma_j)
    :=
    \left\{
    Q:\ W_1(\widehat P_j,Q)
    \le
    r_j(n_j,\gamma_j,\sigma)
    \right\},
\]
where
\[
    r_j(n_j,\gamma_j,\sigma)
    :=
    \frac{512\sigma}{\sqrt{n_j}}
    +
    \sigma
    \sqrt{
    \frac{256e}{n_j}
    \log\frac{1}{1-\gamma_j}
    }.
\]
Then $\mathcal{C}_j^{W_1}(\widehat P_j,\gamma_j)$ is a
$\gamma_j$-confidence set for $P_j$ by the concentration inequality of
\citet{wasser_dist}.

For the Wasserstein discrepancy
$\Loss(Q,\widehat Q_j)=W_1(Q,\widehat Q_j)$, the triangle inequality gives
\[
    \sup_{Q\in\mathcal{C}_j^{W_1}(\widehat P_j,\gamma_j)}
    W_1(Q,\widehat Q_j)
    \le
    W_1(\widehat P_j,\widehat Q_j)+r_j(n_j,\gamma_j,\sigma).
\]
Thus one may use the conservative pseudo-discrepancy
\[
    \hat\Delta_j^{W_1}
    :=
    W_1(\widehat P_j,\widehat Q_j)+r_j(n_j,\gamma_j,\sigma),
\]
which preserves the calibration guarantee because it upper-bounds the
worst-case discrepancy over the confidence set.
\end{example}

\section{Additional Simulator Applications}\label{apx:add_app}

\subsection{EEDI Dataset}

Our dataset is EEDI~\cite{he2024psychometric}, built on the NeurIPS 2020 Education Challenge~\cite{pmlr-v133-wang21a}, which consists of student responses to mathematics multiple-choice questions collected on the Eedi online education platform. The full corpus includes 573 distinct questions and 443{,}433 responses from 2{,}287 students, and each question has four options A-D that we binarize as ``correct/incorrect'' based on the student's or simulator's choice, consistent with Lemma~\ref{CH-bin}. We adopt the preprocessed version curated by \cite{huang2025uncertaintyquantificationllmbasedsurvey}, which retains questions with at least 100 student responses and excludes items with graphs or diagrams, yielding 412 questions. EEDI also provides individual-level covariates such as gender, age, and socioeconomic status, which the authors of \cite{huang2025uncertaintyquantificationllmbasedsurvey} use to construct synthetic profiles. Under the same problem formulation, they compute $\{\hat p_j,\hat q_j\}_{j=1}^{412}$ for seven LLMs: \textsc{GPT-3.5-Turbo} (\texttt{gpt-3.5-turbo}), \textsc{GPT-4o} (\texttt{gpt-4o}), and \textsc{GPT-4o-mini} (\texttt{gpt-4o-mini}); \textsc{Claude 3.5 Haiku} (\texttt{claude-3-5-haiku-20241022}); \textsc{Llama 3.3 70B} (\texttt{Llama-3.3-70B-Instruct-Turbo}); \textsc{Mistral 7B} (\texttt{Mistral-7B-Instruct-v0.3}); \textsc{DeepSeek-V3} (\texttt{DeepSeek-V3}), and constructed a benchmark random simulator that selects uniformly among the available answer choices. A more detailed exploration into the EEDI dataset and the simulation procedure can be found in \cite{huang2025uncertaintyquantificationllmbasedsurvey}.

We apply our methodology to produce a fidelity profile for each candidate LLM. We use absolute error as the loss, \(\Loss(p,q)=|p-q|\) and set the simulation budget \(k=50\). We set the adaptive coverage to be $\gamma_j = 1 - n_j^{-1/2}$ in this example.

\begin{figure}[ht!]
    \centering
    \includegraphics[width=0.7\linewidth]{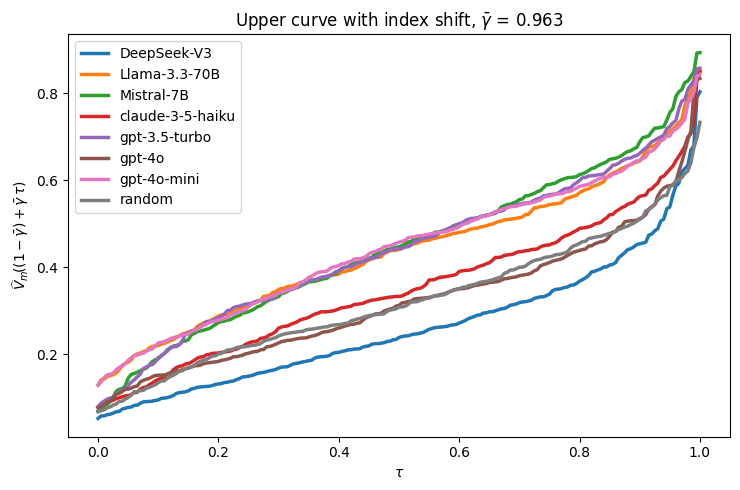}
    \caption{Quantile fidelity profiles \(\hat V(\tau)\) across LLMs (Discrepancy: Absolute loss, \(k=50\)).}
    \label{fig:fidelity_profile}
\end{figure}

Figure~\ref{fig:fidelity_profile} compares models by how tightly their synthetic outcomes track the human distribution across items. We plot 
$\hat V_{\ell}(\tau)$ against $\tau$, where lower-flatter curves indicate uniformly small discrepancies, while elbows reveal rare but severe misses. \textsc{DeepSeek-V3} lies lowest across most quantiles, indicating the most reliable alignment, with the random benchmark and \textsc{GPT-4o} close behind. Notably, several models do not outperform the random baseline, suggesting they may be ill-suited for agent-based simulation under this discrepancy function.

\subsection{OpinionQA}

Our dataset is OpinionQA~\cite{llm_opinion}, built from the Pew Research's American Trends Panel, which consists of the US population's responses to survey questions spanning topics such as racial equity, security, and technology. We adopt the preprocessed version curated by \cite{huang2025uncertaintyquantificationllmbasedsurvey}, which includes 385 distinct questions and 1{,}476{,}868 responses from at least 32{,}864 people. Each question has 5 choices, corresponding to the order sentiments, which is a multinomial setting. We can construct confidence sets $\mathcal C_j$ for multinomial vectors by adopting Example~\ref{ex:multinom} with $d=5$. OpinionQA also provides individual-level covariates such as gender, age, socioeconomic status, religious affiliation, and marital status, and more, which are used to construct synthetic profiles. Under the same problem formulation, the authors of \cite{huang2025uncertaintyquantificationllmbasedsurvey} compute $\{\hat p_j,\hat q_j\}_{j=1}^{385}$ for seven LLMs: \textsc{GPT-3.5-Turbo} (\texttt{gpt-3.5-turbo}), \textsc{GPT-4o} (\texttt{gpt-4o}), and \textsc{GPT-4o-mini} (\texttt{gpt-4o-mini}); \textsc{Claude 3.5 Haiku} (\texttt{claude-3-5-haiku-20241022}); \textsc{Llama 3.3 70B} (\texttt{Llama-3.3-70B-Instruct-Turbo}); \textsc{Mistral 7B} (\texttt{Mistral-7B-Instruct-v0.3}); \textsc{DeepSeek-V3} (\texttt{DeepSeek-V3}), and constructed a baseline random simulator that selects uniformly among the available answer choices. A more detailed exploration into the OpinionQA dataset and the simulation procedure can be found in \cite{huang2025uncertaintyquantificationllmbasedsurvey}. 

We apply our methodology to produce a fidelity profile for each candidate LLM. We use total variation as the discrepancy measure, \(\Loss(p,q) = \tfrac12\,\|p-q\|_1\). In addition, we set the simulation budget \(k=100\), the result is presented in Figure \ref{fig:fidelity_profile_multi}.

\begin{figure}[ht!]
    \centering
    \includegraphics[width=0.9\linewidth]{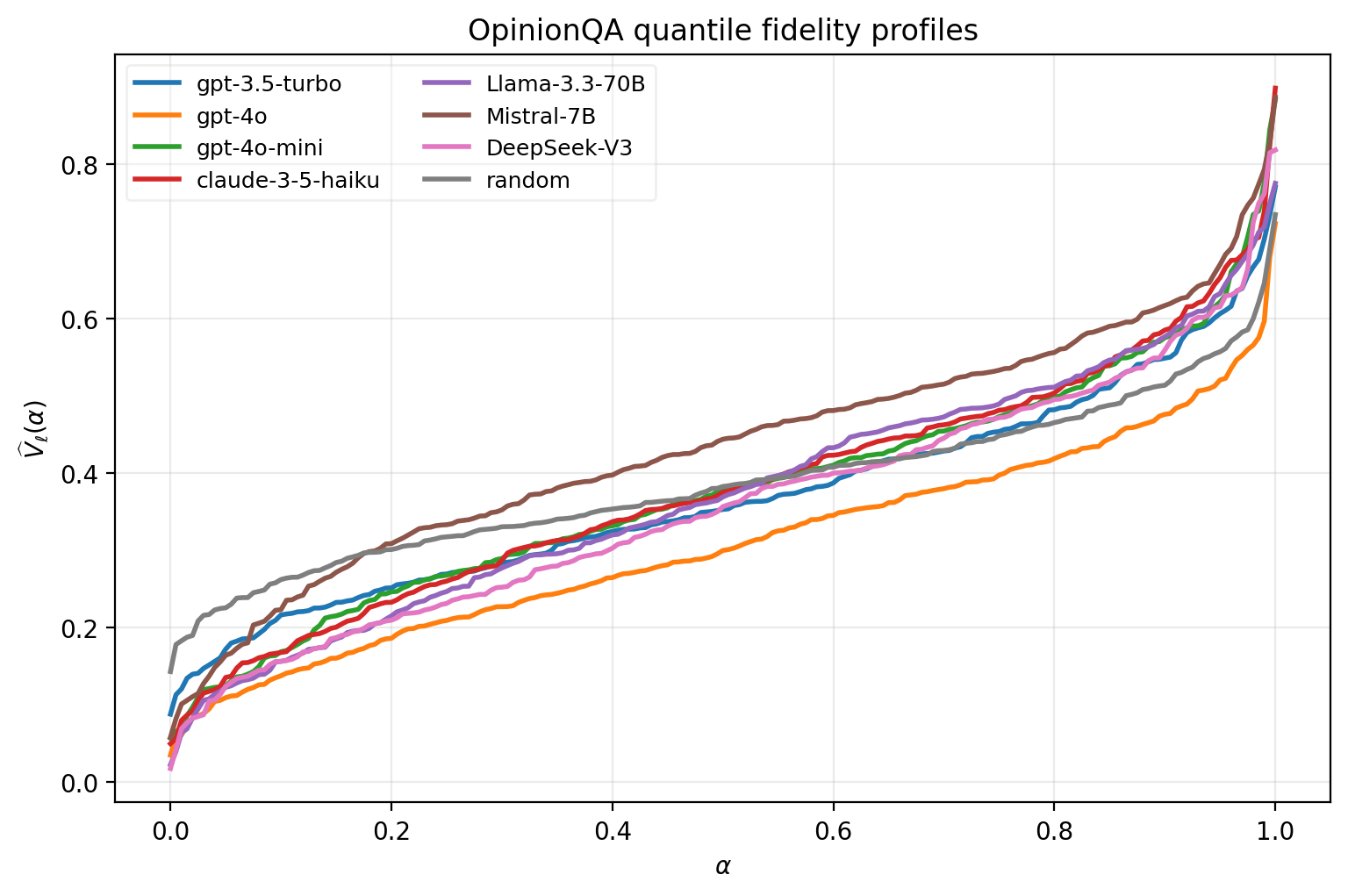}
    \caption{Quantile fidelity profiles \(\hat V(\tau)\) across LLMs.}
    \label{fig:fidelity_profile_multi}
\end{figure}

Figure~\ref{fig:fidelity_profile_multi} compares models by how tightly their synthetic outcomes track the human distribution across items. We plot 
$\hat V_{\ell}(\tau)$ against $\tau$, where lower-flatter curves indicate uniformly small discrepancies, while elbows reveal rare but severe misses. \textsc{GPT-4o} lies lowest across most quantiles, indicating the most reliable alignment. Notably, the simulator curves are steeper than the random benchmark, indicating question-dependent alignment and less uniform discrepancies. This suggests the simulators may require further fine-tuning to achieve more uniform discrepancy levels across this set of questions.

\subsection{UMAR Dataset}

We consider the UMAR unit of NEST (Next Evolution in Sustainable Building Technologies), a modular living-lab platform for building research operated by Empa~\cite{Richner2018NEST}, which is released by \cite{Heer2024ComprehensiveEnergyDemand}. The published dataset contains four years of one-minute measurements, from July 1, 2019 to June 30, 2023, and includes room temperatures, HVAC setpoints, valve and window states, ambient temperature, solar irradiation, and building-energy variables. In our reproduction, the UMAR portion is stored as four yearly wide-format CSV files together with a metadata table for the UMAR measurement channels. 

Following the fixed experimental design used in our UMAR reproduction, we focus on room \texttt{temp\_275} and aggregate the one-minute measurements to 30-minute means. We then construct a prediction table using the current room temperature, two lagged temperatures, hour of day, weekend indicator, month, ambient temperature, and, when available, room setpoint and solar irradiation as covariates. To avoid clear sensor artifacts and outliers, we remove rows for which \texttt{room\_temp\_t}, \texttt{lag1}, \texttt{lag2}, or the future response exceeds \(30^\circ\)C. This yields 66,922 cleaned 30-minute observations.

Unlike EEDI and OpinionQA, UMAR is not a multinomial-response setting. Here the unit of analysis is an operational scenario rather than a question. We define the scenario label as
\[
\psi = (h,w,t_q),
\]
where \(h \in \{0,\dots,23\}\) is the hour of day, \(w \in \{0,1\}\) indicates whether the timestamp falls on a weekend, and \(t_q \in \{0,1,2,3\}\) is the quartile of the ambient temperature. This gives \(m=192\) scenarios in total. For each scenario \(\psi_j\), the real-side parameter \(p_j\) is the variance of the next-step 30-minute room temperature among held-out UMAR observations in that scenario, while the simulator-side parameter \(q_j\) is the corresponding variance induced by emulator draws. Since this is a continuous-outcome problem, we construct the confidence set \(\mathcal C_j\) for \(p_j\) using a bootstrap confidence interval for the scenario-level variance, with confidence level \(\gamma_j = 1 - n_j^{-1/3}\), where \(n_j\) is the number of real observations in scenario \(j\).

Since this is a continuous-outcome problem, we estimate the real-side and simulator-side scenario variances by the unbiased sample variances \(\hat p_j\) and \(\hat q_j\), and use the squared-loss discrepancy $L(p_j,q_j) = (p_j-q_j)^2.$
For the real-side variance parameter \(p_j\), we construct \(\mathcal C_j=[\ell_j,u_j]\) as a percentile bootstrap confidence interval based on \(B=400\) bootstrap resamples of the real observations in scenario \(j\). Writing \(\hat p_j^{*(1)},\ldots,\hat p_j^{*(B)}\) for the bootstrap sample variances and setting \(\alpha_j=(1-\gamma_j)/2\) with \(\gamma_j = 1 - n_j^{-1/3}\), we define $ \mathcal C_j = [\ell_j,u_j] = \left[ Q_{\alpha_j}\!\left(\{\hat p_j^{*(b)}\}_{b=1}^B\right), Q_{1-\alpha_j}\!\left(\{\hat p_j^{*(b)}\}_{b=1}^B\right) \right].$
The scenario-level upper pseudo-discrepancy used in Figure~\ref{fig:umar_fidelity_profile} is then
\[
\widehat{\Delta}_j^+ = \sup_{u\in\mathcal C_j} L(u,\hat q_j) = \max\{(\ell_j-\hat q_j)^2,(u_j-\hat q_j)^2\}.
\]
We use this bootstrap interval as a practical approximation for variance uncertainty in the continuous-outcome UMAR setting, rather than as a finite-sample-exact confidence set.

We compare three emulator families using five-fold cross-fitting: linear regression model, decision tree, and multilayer perceptron. For each emulator and each scenario, we generate a balanced simulator budget of \(k=200\) draws. We then apply our quantile-fidelity procedure to the resulting scenario-level variance discrepancies.

\begin{figure}[ht!]
    \centering
    \includegraphics[width=0.78\linewidth]{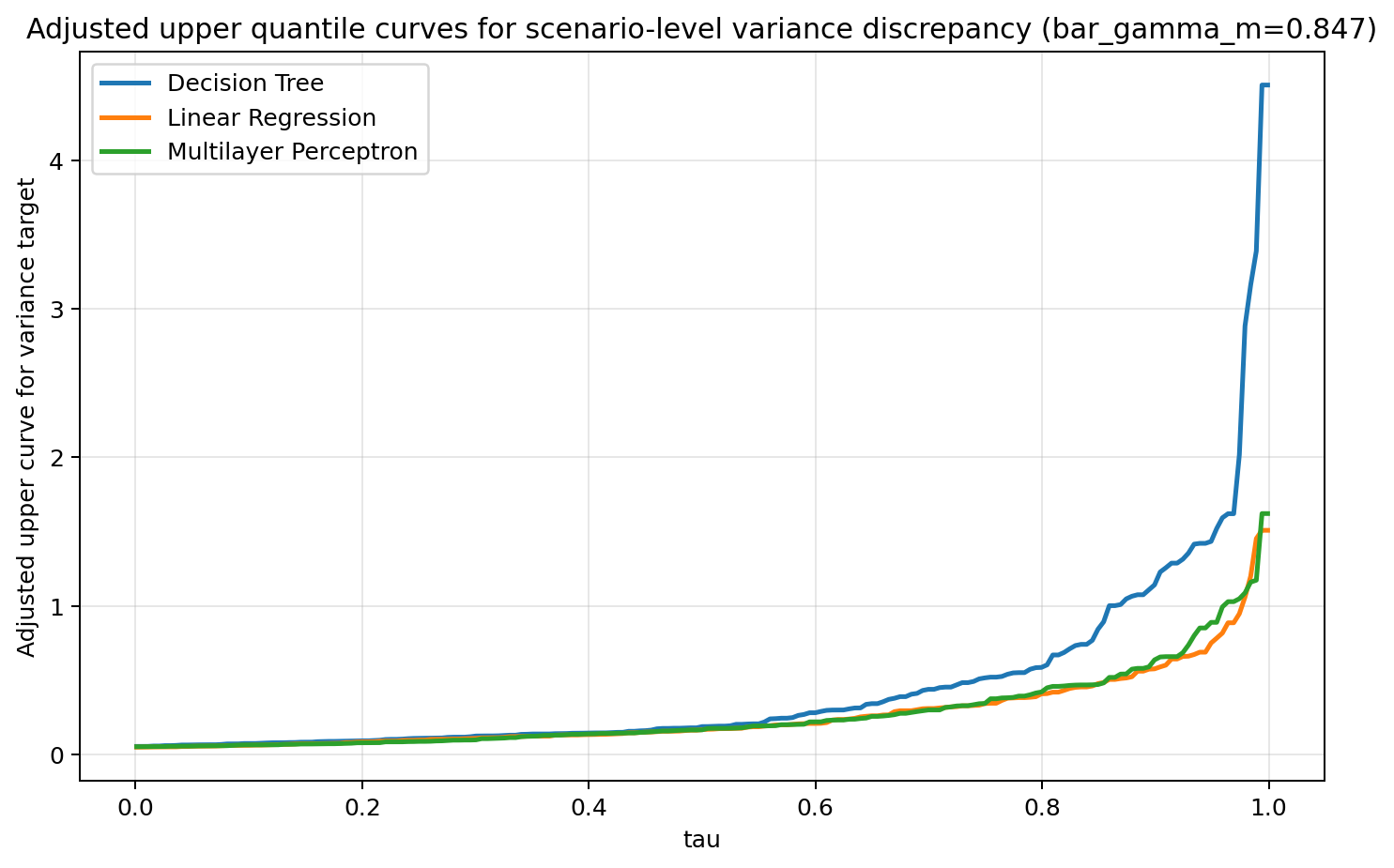}
    \caption{Adjusted upper quantile curves for scenario-level variance discrepancy across UMAR emulator families. Lower curves indicate closer agreement between emulator-side and real-side temperature variability.}
    \label{fig:umar_fidelity_profile}
\end{figure}

Figure~\ref{fig:umar_fidelity_profile} reports the adjusted upper quantile curves for the UMAR experiment. Lower and flatter curves indicate that the emulator more faithfully reproduces the scenario-level variability observed in the real building data. In this reproduction, the linear regression model attains the lowest overall upper quantile curve, with the multilayer perceptron close behind, while the decision-tree baseline exhibits substantially larger upper-tail discrepancies. This illustrates that predictive accuracy alone is not the main object of interest here: the fidelity target is the distribution of scenario-level variability, not only pointwise temperature prediction error.

\end{document}